\documentclass[lettersize,journal]{IEEEtran}
\usepackage{amsmath,amsfonts}
\usepackage{algorithmic}
\usepackage{algorithm}
\usepackage{array}
\usepackage[caption=false,font=normalsize,labelfont=sf,textfont=sf]{subfig}
\usepackage{textcomp}
\usepackage{stfloats}
\usepackage{url}
\usepackage{verbatim}
\usepackage{graphicx}
\usepackage{cite}
\usepackage{orcidlink}
\usepackage{bm}
\usepackage{amsthm}
\newtheoremstyle{remarkstyle}
{3pt}
{3pt}
{\normalfont}
{\parindent}
{\bfseries\itshape}
{:}
{.5em}
{\bfseries\itshape\thmname{#1}\thmnumber{ #2}}

\theoremstyle{remarkstyle}
\newtheorem{remark}{Remark}
\usepackage{booktabs}

\newtheoremstyle{problemstyle}
{3pt}
{3pt}
{\normalfont}
{\parindent}
{\bfseries\itshape}
{:}
{.5em}
{\bfseries\itshape\thmname{#1}\thmnumber{ #2}}

\theoremstyle{problemstyle}

\usepackage{booktabs}

\usepackage{amsthm}

\newtheoremstyle{mylemma}
{0.5\topsep}
{0.5\topsep}
{\normalfont}
{}
{\itshape}             
{: }
{ }
{}

\theoremstyle{mylemma}
\newtheorem{lemma}{Lemma}

\usepackage{amsmath}
\usepackage{amssymb}
\usepackage{subcaption}

\usepackage{bm}
\newtheoremstyle{observation}
{3pt}
{3pt}
{\itshape}
{0.5em}
{\bfseries\itshape}
{: }
{ }
{}

\theoremstyle{observation}
\newtheorem{observation}{\textbf{\textit{}}}

\usepackage{float}
\usepackage{multirow}
\usepackage{subcaption}

\begin{document}

\title{Cross-channel Specific Emitter Identification and Verification \\ via Signal Envelope}
\author{Yuhao Chen\orcidlink{0009-0000-5050-1688},
	Boxiang He\orcidlink{0000-0002-9235-1144},
	Shilian Wang\orcidlink{0000-0003-4132-8750},
	and Jing Lei\orcidlink{0000-0002-5838-5826}
	\thanks{Y. Chen, B. He, S. Wang, and J. Lei are with the College of Electronic Science and Technology, National University of Defense Technology, Changsha 410003, P. R. China. (email: cyh20220720@163.com; boxianghe1@bjtu.edu.cn;wangsl@nudt.edu.cn;leijing@nudt.edu.cn).}
}
\markboth{}%
{Shell \MakeLowercase{\textit{et al.}}: Robust Physical Layer Authentication for Time-Varying Channels via Feature Decoupling and Diffusion Models}


\maketitle
\begin{abstract}
	Specific emitter identification (SEI) determines which known emitter a received signal originates from, while specific emitter verification (SEV) determines whether the received signal genuinely comes from its claimed emitter. In this paper, we consider the effect of wireless fading channels on SEI and SEV. When the Rician $K$-factor varies, the resulting distribution shift induced by the channel degrades both identification and verification performance. To address this issue, we first theoretically prove that the coefficient of variation of the signal envelope is strictly monotonic with respect to the Rician $K$-factor. Motivated by this property, we propose an envelope-guided adaptive feature modulation (EAFM) identifier for SEI and an EAFM with Mahalanobis distance metric learning (EAFM-MD) verifier for SEV. Specifically, the proposed EAFM identifier adopts a dual-branch neural network to extract device-oriented features from the IQ-domain input and channel-conditioning features from the normalized signal envelope, and adaptively modulates the former via feature-wise linear modulation. Then, we extend the EAFM identifier to an EAFM-MD verifier. The device-fingerprint library is constructed by storing the feature centroid and covariance for each enrolled device, along with the within-device Mahalanobis distances of training signals. For verification, the Mahalanobis distance between the extracted test features and each stored centroid is computed using the stored covariance matrix, and the minimum distance is compared to the corresponding device threshold to make a decision. Finally, numerical results show that the proposed EAFM identifier improves cross-channel identification performance, while the proposed EAFM-MD verifier achieves superior detection performance against unknown spoofing attacks.
\end{abstract}


\section{Introduction}
\IEEEPARstart{S}{pecific} emitter identification (SEI) determines which known emitter a received signal originates from, while specific emitter verification (SEV) validates whether a received signal genuinely comes from its claimed emitter~\cite{refTalbotSEI}. Both techniques leverage radio-frequency fingerprints (RFFs), which are physical-layer features formed by inherent hardware differences and are characterized by uniqueness, unforgeability, and cost advantages\cite{Dinh2025Radio,Zhang2025Radio,Zhang2023Radio}. Due to their stability and resistance to cloning, RFFs enable key-free, lightweight device identification and defense against spoofing attacks in physical-layer secure communication systems~\cite{Xie2022Multiple}. 

As illustrated in Fig.~\ref{problem_SEIV}, SEV and SEI can be integrated into a unified framework for practical deployment. Within this framework, SEV determines whether a received signal originates from an enrolled legitimate device, while SEI identifies the specific device identity. 
Notably, SEI and SEV have been applied in practical scenarios such as unmanned aerial vehicle (UAV) communication~\cite{Teng2024Exploiting} and industrial internet of things (IIoT)~\cite{Zhou2025Radio}\cite{Shen2022Towards} for device identification and identity verification.
However, a critical challenge in practical SEI and SEV deployment is the fundamental mismatch between the training and testing environments in terms of channel conditions~\cite{refNeedleHaystack}\cite{refElmaghbubIQ}. In real-world deployment scenarios, such as UAVs operating across diverse environments including urban areas, open fields, and mountainous regions, channel conditions can vary dramatically~\cite{Wang2024Robust}. When models trained under specific channel conditions are directly applied to significantly different channel environments, their performance degrades substantially~\cite{Wang2025Avoiding}. This degradation is closely related to specific channel parameters, such as the Rician-$K$ factor~\cite{Li2024Dataset}.

\begin{figure}[t]
	\centering
	\includegraphics[width=0.45\textwidth]{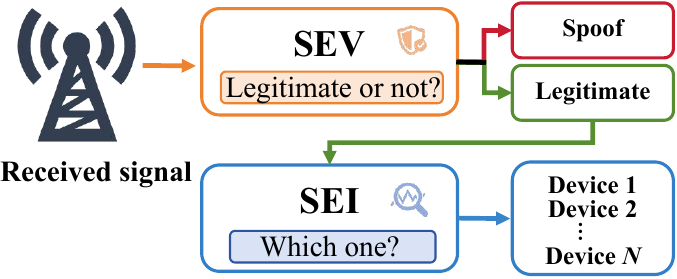}
	\caption{Illustration of SEV for legitimacy verification followed by SEI for device identification.}
	\label{problem_SEIV}
\end{figure}

To address this challenge, various approaches focusing on signal-level and model-level improvements have been proposed~\cite{refC0330}. At the feature and model levels, researchers train RFF extractors using deep metric learning and improve their performance by modeling multipath and Doppler effects via power-delay profiles and by collecting samples across various scenarios, including static line-of-sight (LoS) and randomly generated channel effects~\cite{Shen2022Towards}. However, collecting sufficient training data to capture all possible characteristics of the wireless channel environment is challenging~\cite{refHe}\cite{refNguyenTL}. In~\cite{refXieRF}, a disentangled representation learning approach is proposed that uses adversarial learning to separate signals into device-dependent and device-independent components. The training data are mainly collected under LoS propagation, which eliminates the need for additional datasets.

Additionally, some researchers leverage power amplifier nonlinearity to mitigate channel effects and adopt transfer learning to reduce the need for large amounts of labeled training data while quickly adapting to different channel conditions~\cite{Yang2024RobustnessSecurity}. As confirmed in~\cite{Li2024Dataset}, a low Rician-$K$ factor in UAV communications reduces RFF identification accuracy, and a data augmentation method is proposed to mitigate this issue. Generic out-of-distribution (OOD) generalization methods such as Domain-Adversarial Neural Network (DANN) \cite{Ganin2016Domain} and Variance-Risk Extrapolation (VREx) \cite{Krueger2021OutOfDistribution} have shown promise in other domains, but their application to cross-channel RFF-SEI remains sub-optimal due to the lack of explicit encoding of channel-conditioning information. In summary, SEI and SEV may remain challenging when the underlying channel statistics shift substantially between the training and deployment phases.

In this paper, we propose an envelope-guided adaptive feature modulation (EAFM) identifier for SEI to enhance RFF robustness against channel variations. The EAFM identifier utilizes a dual-branch neural network with EAFM via feature-wise linear modulation (FiLM)~\cite{Perez2018FiLM}. Specifically, the identifier leverages signal-envelope information to generate channel-conditioning features that adaptively modulate the device-oriented features via FiLM, thereby mitigating channel-induced distortions and improving classification performance across varying wireless fading conditions.

Furthermore, in deployment, received signals may originate from unknown spoofing devices rather than legitimate transmitters, which requires SEV to detect such attacks~\cite{refTalbotSEI}\cite{refHiguero}. Unlike SEI, which identifies known emitters, SEV must distinguish unknown transmitters from enrolled legitimate ones. Softmax classifiers are not inherently designed for such a verification task~\cite{refHiguero}, so we further propose an EAFM with Mahalanobis Distance metric learning (EAFM-MD) verifier for SEV. The proposed EAFM-MD verifier adopts a fingerprint library-based approach~\cite{refFingerBLE}\cite{refChilletRiFyFi} using Mahalanobis Distance metric learning, inspired by advances in meta-learning~\cite{Liu2020FewShot}. Specifically, the EAFM-MD verifier represents each enrolled device as a statistical distribution in the feature space and verifies devices by computing Mahalanobis distances to these distributions. Importantly, this verification capability is achieved without requiring any attack samples during training, which allows the system to detect unknown spoofing devices.

To validate our methods, we evaluate performance over Rician fading channels with a wide range of \(K\)-factors, from high-\(K\) LoS scenarios to low-\(K\) NLoS conditions, covering typical and challenging environments for UAV communications. The main contributions of this paper are as follows:

\begin{itemize}
	\item We theoretically prove that the coefficient of variation of the signal envelope is strictly monotonic with respect to the Rician $K$-factor, which provides a theoretical foundation for envelope-guided feature modulation.
	
	\item We propose an EAFM identifier for SEI with cross-channel OOD. The EAFM identifier uses a dual-branch neural network to extract device-oriented features from the IQ-domain input and channel-conditioning features from the normalized signal envelope, where the latter adaptively modulates the former via feature-wise linear modulation. Numerical results show that the proposed EAFM identifier achieves superior identification accuracy under mismatched Rician $K$-factor conditions.
	
	\item Building on the EAFM identifier, we further propose an EAFM-MD verifier for SEV. The EAFM-MD verifier learns device-specific prototypes and precision matrices and constructs a fingerprint library using the Mahalanobis distance, so that the detection of unknown spoofing devices is possible without any attack samples during training. Numerical results show that the proposed EAFM-MD verifier achieves superior detection probability and precision in OOD scenarios with unknown attacks.
\end{itemize}

\emph{Notation:} Throughout this paper, scalars, vectors, and matrices are denoted by lower-case italic letters $x$, bold lower-case italic letters $\bm{x}$, and bold capital italic letters $\bm{X}$, respectively. A random variable and its realization are respectively written as $\mathsf{x}$ and $x$. The operator $[\cdot]^{\mathsf{T}}$ denotes the transpose. The symbols $\odot$ and $\log(\cdot)$ denote the Hadamard product and logarithm, respectively. $\Re\{\cdot\}$ and $\Im\{\cdot\}$ represent the real and imaginary parts of a complex number, respectively. The notation $\mathcal{CN}\left({\mu},\sigma^2\right)$ denotes the probability density function of a random variable following the complex Gaussian distribution with mean $\mu$ and variance $\sigma^2$. $\min(\cdot,\cdot)$ represents the minimum value function. Let $j=\sqrt{-1}$. $\mathbb{E}\{\cdot\}$ denotes the expectation operator with respect to all random variables.

\begin{figure}[t]
	\centering
	\includegraphics[width=0.475\textwidth]{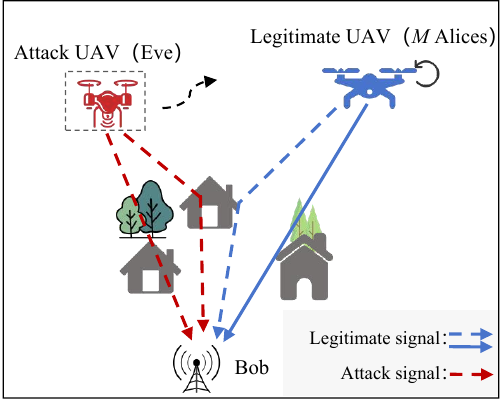}
	\caption{System scenario of UAV communications in environment-dependent Rician fading channels with potential spoofing attacks. Bob receives signals from $M$ legitimate UAVs, while an unknown attack UAV may exist and attempt to spoof legitimate identities.}
	\label{Proposed_imSMEI_nomp}
\end{figure}

\section{System Model and Problem Formulation}
\label{chapter:system_model_and_problem}
In this section, we first describe the SEI system and the communication scenario. Then, we formulate the SEI problem under two distinct settings.
\subsection{System model}
\label{system_model_sub}

We consider an SEI system in a UAV communication as illustrated in Fig.~\ref{Proposed_imSMEI_nomp}, where $M$ legitimate transmitters indexed by $\mathcal{I}_M = \{1, 2, \ldots, M\}$ communicate with a receiver (Bob). However, an unknown attack transmitter (Eve) may exist and attempt to spoof legitimate identities by impersonating a known device. Denote the active legitimate transmitter as Alice~$m$, $m \in \mathcal{I}_M$. The baseband equivalent received signal from Alice~$m$ can be represented as
\begin{align}
	y_m(t) &= h_m \, d_m\!\left(x_m(t)\right) + w(t), \label{eq:received_signal}\\
	h_m &= \sqrt{\frac{K}{K+1}} \, e^{j\theta_m} \;+\; \sqrt{\frac{1}{K+1}} \, \tilde{h}_m, \label{eq:rician_channel}
\end{align}
where $x_m(t)$ is the original transmitted signal of Alice~$m$; $h_m$ is the channel coefficient from Alice~$m$ to Bob, which comprises a dominant line-of-sight (LoS) component and scattered multipath components; $K$ is the Rician $K$-factor representing the power ratio between the LoS and scattered components; $\theta_m$ is the phase of the dominant component; $\tilde{h}_m$ captures the scattered multipath fading; $d_m(\cdot)$ is the RF distortion function of Alice~$m$; $w(t) \sim \mathcal{CN}(0,\sigma^2)$ is the additive white Gaussian noise.

Following the established modeling framework in~\cite{He2020CooperativeSEI}, the device-specific distortion function $d_m\left(\cdot\right)$ is characterized by a cascade of hardware impairments inherent to each transmitter's RF front-end. The distortion process begins at the I/Q modulator, where phase and gain mismatches between the in-phase and quadrature branches introduce characteristic fingerprints. The signal after I/Q imbalance is expressed as
\begin{equation}
	\label{eq:iq_imbalance}
	x_m'\left(t\right) = \mu_m x_m\left(t\right) + \nu_m x_m^*\left(t\right),
\end{equation}

\noindent where $x_m^*\left(t\right)$ denotes the complex conjugate, and $\mu_m$ and $\nu_m$ are device-specific modulator distortion parameters defined as
\begin{align}
	\mu_m &= \tfrac{1}{2}\left(G_m+1\right)\cos\!\left(\tfrac{\zeta_m}{2}\right) + j\,\tfrac{1}{2}\left(G_m-1\right)\sin\!\left(\tfrac{\zeta_m}{2}\right), \label{eq:mu_m} \\
	\nu_m &= \tfrac{1}{2}\left(G_m-1\right)\cos\!\left(\tfrac{\zeta_m}{2}\right) + j\,\tfrac{1}{2}\left(G_m+1\right)\sin\!\left(\tfrac{\zeta_m}{2}\right), \label{eq:nu_m}
\end{align}

\noindent where $G_m$ denotes the gain imbalance and $\zeta_m$ is the phase bias. Following the I/Q imbalance, the signal undergoes frequency up-conversion, during which carrier leakage and spurious tones are introduced. The distorted signal can then be given as
\begin{equation}
	\label{eq:carrier_spurious}
	x_m''\left(t\right) = \left(x_m'\left(t\right) + \xi_m\right) e^{j2\pi f_c t} + a_m^{\mathrm{ST}} e^{j2\pi\left(f_c + f_m^{\mathrm{ST}}\right)t},
\end{equation}

\noindent where $f_c$ is the carrier frequency; $\xi_m$ denotes the carrier leakage coefficient; $a_m^{\mathrm{ST}}$ and $f_m^{\mathrm{ST}}$ are the amplitude and frequency offset of the spurious tone, respectively.

Finally, the signal passes through the power amplifier (PA), whose nonlinear distortion is modeled by a Taylor series expansion, i.e.
\begin{align}
	x_m'''\left(t\right) &= \sum_{l=1}^{L} b_{m,l} \left(x_m''\left(t\right)\right)^l, \label{eq:pa_nonlinearity} \\
	&= d_m\!\left(x_m\left(t\right)\right), \label{eq:total_distortion}
\end{align}

\noindent where $L$ is the Taylor polynomial order and $b_{m,l}$ are the device-specific PA coefficients. The complete distortion function $d_m\left(\cdot\right)$ in \eqref{eq:total_distortion} encapsulates all four hardware impairments with device-specific parameters $\left\{G_m, \zeta_m, \xi_m, a_m^{\mathrm{ST}}, f_m^{\mathrm{ST}}, b_{m,l}\right\}$.

\begin{figure}[t]
	\centering
	\includegraphics[width=0.45\textwidth]{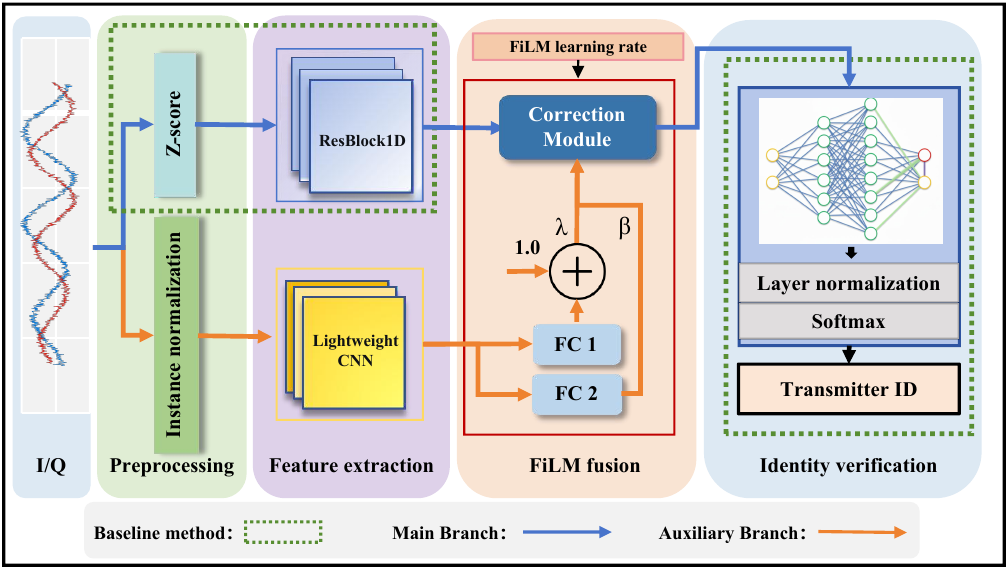}
	\caption{The architecture of the proposed EAFM identifier for closed-set SEI. 
		The dual-branch design processes IQ signals (main branch) 
		and normalized envelope (auxiliary branch) separately, 
		with FiLM-based modulation dynamically adjusting 
		main features through envelope guidance.}
	\label{Proposed_EAFM}
\end{figure}

\subsection{Problem formulation}
\label{problem_formulation}

The goal of SEI is to identify the true transmitter identity from received signals that are affected by both hardware distortions and wireless channel effects. Here, we consider the SEI problem as a mapping from the received signal matrix to a device identity under two distinct settings as
\begin{align}
	\widehat{m}_i^{\mathrm{c}} &= \mathcal{G}_{\mathrm{c}}\left(\bm{Y}_i\right), \quad \widehat{m}_i^{\mathrm{c}} \in \mathcal{I}_M, \label{eq:closed_set_sei} \\
	\widehat{m}_i^{\mathrm{o}} &= \mathcal{G}_{\mathrm{o}}\left(\bm{Y}_i\right), \quad \widehat{m}_i^{\mathrm{o}} \in \mathcal{I}_M \cup \left\{\text{Eve}\right\}, \label{eq:open_set_sei}
\end{align}
\noindent where $\mathcal{G}_{\mathrm{c}}\left(\cdot\right)$ and $\mathcal{G}_{\mathrm{o}}\left(\cdot\right)$ denote the identifiers for the closed-set SEI and the open-set SEV settings, respectively. For the closed-set setting, $\widehat{m}_i^{\mathrm{c}}$ is the estimated transmitter label belonging to the legitimate device set $\mathcal{I}_M$. For the open-set setting, $\widehat{m}_i^{\mathrm{o}}$ is the predicted identity that either corresponds to a legitimate device in $\mathcal{I}_M$ or detects an unknown Eve. $\bm{Y}_i \in \mathbb{R}^{2 \times N}$ is the $i$-th received signal matrix converted from the received signal $y_m\left(t\right)$ after energy normalization and global standardization, i.e.
\begin{equation}
	\label{eq:input_format}
	\bm{Y}_i = \begin{bmatrix}
		\Re\left\{r_i\left(1\right)\right\} & \Re\left\{r_i\left(2\right)\right\} & \cdots & \Re\left\{r_i\left(N\right)\right\} \\
		\Im\left\{r_i\left(1\right)\right\} & \Im\left\{r_i\left(2\right)\right\} & \cdots & \Im\left\{r_i\left(N\right)\right\}
	\end{bmatrix}, \quad i \in \mathcal{I}_S,
\end{equation}
\noindent where $\mathcal{I}_S = \left\{1, 2, \ldots, S\right\}$ is the index set of $S$ received signal segments; $r_i\left(n\right)$ denotes the $n$-th discrete-time sample at the symbol rate; $N$ is the number of samples per segment.

\section{Proposed EAFM identifier for SEI}
\label{chapter:3}

In this section, we propose the EAFM identifier to address the SEI problem under varying Rician $K$-factor conditions. We first establish in Lemma~\ref{lemma1} that the coefficient of variation (CV) of the signal envelope is strictly monotonic with respect to $K$, providing a theoretical basis for using envelope statistics as channel-conditioning information. Motivated by this finding, we design a dual-branch architecture with three cascaded modules, as illustrated in Fig.~\ref{Proposed_EAFM}. The optimization problem for the proposed EAFM identifier is then formulated.

To characterize the relative fluctuation of the envelope for emitter identification, we introduce the CV of the Rician $K$-factor~\cite{Arachchige2022Robust,Hidalgo2022Amplitude}, i.e.
\begin{equation}
	C_v\left(K\right)=\frac{\sqrt{\operatorname{Var}\!\left[R\right]}}{\mathbb{E}\!\left\{R\right\}},
	\label{eq:cv-def}
\end{equation}
\noindent where $R$ denotes the envelope amplitude, $\mathbb{E}\left\{\cdot\right\}$ is the expectation operator, and $\operatorname{Var}\left[\cdot\right]$ represents the variance. The following lemma shows how this function varies with respect to $K$.

\begin{lemma}
	\label{lemma1}
	For the Rician distribution, $C_v(K)$ is a strictly decreasing function of $K$, i.e.
	\begin{equation}
		\frac{\partial C_v(K)}{\partial K} < 0,
		\label{eq:cv-derivative-neg}
	\end{equation}
	\noindent where $\frac{\partial C_v(K)}{\partial K}$ denotes the derivative in $K$.
\end{lemma}
\begin{proof}
	See Appendix~\ref{appendix:fano_derivation}.
\end{proof}

Lemma~\ref{lemma1} establishes that the statistical fluctuation pattern of the normalized envelope can serve as measurable channel-conditioning information, thereby motivating the use of envelope statistics to guide adaptive feature correction.

\subsection{EAFM identifier architecture}
\label{eafm_framework}
The proposed EAFM is designed to implement the identifier $\mathcal{G}_{\mathrm{c}}(\cdot,\cdot)$ for achieving robust device identification under varying Rician $K$-factor conditions. Building on the system model in Section~\ref{chapter:system_model_and_problem}, the proposed EAFM identifier is a composition of three cascaded modules: a dual-branch feature extraction module, a FiLM correction module, and an identity classifier module. The identifier maps the IQ-domain input and the normalized envelope to a predicted transmitter label $\widehat{m}_i^{\mathrm{c}}$, which is given by the composition of the three cascaded modules, i.e.
\begin{align}
	\widehat{m}_i^{\mathrm{c}} &= \mathcal{G}_{\mathrm{c}}\left(\bm{Y}_i,\, \hat{\bm{r}}_i\right) \label{eq:eafm_composition_a}, \\
	&= \operatorname*{arg\,max}_{m\in\mathcal{I}_M} \; \Big[ \mathcal{C}\bigl( \mathcal{F}\bigl( \varPhi_{\text{main}}\left(\bm{Y}_i\right),\, \varPhi_{\text{aux}}\left(\hat{\bm{r}}_i\right) \bigr) \bigr) \Big]_m \label{eq:eafm_composition_b},
\end{align}
\noindent where $\hat{\bm{r}}_i \in \mathbb{R}^{N}$ is the normalized envelope of the complex baseband sequence $\{r_i(n)\}_{n=1}^{N}$. $\varPhi_{\text{main}}(\cdot)$ and $\varPhi_{\text{aux}}(\cdot)$ constitute the dual-branch feature extraction module that map $\bm{Y}_i$ and $\hat{\bm{r}}_i$ to device-oriented and channel-conditioning features, respectively. $\mathcal{F}(\cdot,\cdot)$ denotes the FiLM correction module that adaptively modulates the device-oriented features under the guidance of the channel-conditioning features. $\mathcal{C}(\cdot)$ represents the identity classifier module that outputs the predicted logits. 

Building upon the composition in \eqref{eq:eafm_composition_b}, we then elaborate on the architecture and design rationale of each module.

\paragraph{Dual-branch feature extraction module}
To extract both device-specific and channel-conditioning information from the
same received signal, the feature extraction module adopts a parallel dual-branch
design, where the outputs are given by
\begin{align}
	\bm{f}_{\text{main},i} &= \varPhi_{\text{main}}\left(\bm{Y}_i\right)
	\in \mathbb{R}^{d}, \label{eq:main_feature} \\
	\bm{f}_{\text{aux},i} &= \varPhi_{\text{aux}}\left(\hat{\bm{r}}_i\right)
	\in \mathbb{R}^{d_a}, \label{eq:aux_feature}
\end{align}
\noindent where $\Phi_{\text{main}}\left(\cdot\right)$ extracts the device
fingerprint from the IQ-domain input $\bm{Y}_i$;
$\Phi_{\text{aux}}\left(\cdot\right)$ extracts the envelope signature from $\hat{\bm{r}}_i$;
$d$ and $d_a$ denote the dimensions of the main and auxiliary feature spaces,
respectively.

\paragraph{FiLM correction module}
To dynamically modulate the device-oriented features under the guidance of
channel conditions, the FiLM correction module $\mathcal{F}\left(\cdot,\cdot\right)$
adaptively adjusts the main-branch features using information extracted from the
auxiliary branch, which can be given by
\begin{align}
	\bm{f}_{\text{corr},i} &= \mathcal{F}\left( \bm{f}_{\text{main},i},\, \bm{f}_{\text{aux},i} \right), \label{eq:film_operator}\\
	&= \boldsymbol{\gamma}_i \odot \bm{f}_{\text{main},i} + \boldsymbol{\beta}_i ,\label{eq:film_output}
\end{align}
\noindent where $\odot$ denotes element-wise multiplication, and the
feature-wise scale vector $\boldsymbol{\gamma}_i$ and bias vector
$\boldsymbol{\beta}_i$ are generated from the auxiliary features through two
learnable mappings, which can be expressed as
\begin{align}
	\boldsymbol{\gamma}_i &= \bm{1} + \mathcal{G}_{\gamma}
	\left(\bm{f}_{\text{aux},i}\right) \in \mathbb{R}^{d},
	\label{eq:gamma} \\
	\boldsymbol{\beta}_i &= \mathcal{G}_{\beta}\left(\bm{f}_{\text{aux},i}\right)
	\in \mathbb{R}^{d},
	\label{eq:beta}
\end{align}
\noindent where $\bm{1}$ denotes the all-one vector; $\mathcal{G}_{\gamma}\left(\cdot\right)$ 
and $\mathcal{G}_{\beta}\left(\cdot\right)$ are the learnable linear transformations, 
respectively, with $\mathcal{G}_{\gamma}\left(\cdot\right)$ zero-initialized to preserve 
the identity mapping of the main branch features in early training. Specifically, 
$\mathcal{G}_{\gamma}\left(\cdot\right)$ and $\mathcal{G}_{\beta}\left(\cdot\right)$ can be expressed as
\begin{align}
	\mathcal{G}_{\gamma}\left(\bm{f}_{\text{aux},i}\right) &=
	\bm{W}_{\gamma} \bm{f}_{\text{aux},i} + \bm{b}_{\gamma},
	\label{eq:film_linear_gamma} \\
	\mathcal{G}_{\beta}\left(\bm{f}_{\text{aux},i}\right) &=
	\bm{W}_{\beta} \bm{f}_{\text{aux},i} + \bm{b}_{\beta},
	\label{eq:film_linear_beta}
\end{align}
\noindent where $\bm{W}_{\gamma}, \bm{W}_{\beta} \in \mathbb{R}^{d
	\times d_a}$ are the weight matrices and $\bm{b}_{\gamma}, \bm{b}_{\beta}
\in \mathbb{R}^{d}$ are the bias vectors.

\paragraph{Identity classifier module}
After the FiLM correction, the modulated features $\bm{f}_{\text{corr},i}$ are fed into the identity classifier module $\mathcal{C}\left(\cdot\right)$, which directly maps them to the predicted transmitter label as
\begin{align}
	\widehat{m}_i^{\mathrm{c}} &= \operatorname*{arg\,max}_{m\in\mathcal{I}_M} \; \left[ \mathcal{C}\left( \bm{f}_{\text{corr},i} \right) \right]_m, \label{eq:classifier_output_a} \\
	&= \operatorname*{arg\,max}_{m\in\mathcal{I}_M} \; \left[ \bm{W} \bm{f}_{\text{corr},i} + \bm{b} \right]_m, \label{eq:classifier_output_b}
\end{align}
\noindent where $\bm{W}\in\mathbb{R}^{M\times d}$ and $\bm{b}\in\mathbb{R}^{M}$ are the weight matrix and bias vector of the linear classifier.

\subsection{Optimization of EAFM identifier}
\label{eafm_optimization}

In this subsection, we formulate the training objective for the EAFM identifier. 
Let $\left\{\left(\bm{Y}_i, \hat{\bm{r}}_i, m_i^{\mathrm{c}}\right)\right\}_{i \in \mathcal{I}_S}$ 
denote the training dataset of size $S$, where $m_i^{\mathrm{c}} \in \mathcal{I}_M$ is the device label. 
Thus, the optimization problem can be given by
\begin{equation}
	\boldsymbol{\theta}^{*} = \operatorname*{arg\,min}_{\boldsymbol{\theta}} \; 
	\frac{1}{S} \sum_{i \in \mathcal{I}_S} \mathcal{L}_{\text{CE}}\left( \bm{s}_i,\, m_i^{\mathrm{c}} \right),
	\label{eq:optimization}
\end{equation}
\noindent where $\bm{s}_i = \mathcal{C}\left( \bm{f}_{\text{corr},i} \right)$; $\boldsymbol{\theta}^{*}$ and $\boldsymbol{\theta}$ denote the optimal parameter and the trainable parameters, respectively; 
and the cross-entropy loss is defined as
\begin{equation}
	\mathcal{L}_{\text{CE}}\left( \bm{s}_i,\, m_i^{\mathrm{c}} \right) = -\sum_{m=1}^{M} \mathbb{I}\left(m_i^{\mathrm{c}} = m\right) \log \hat{p}_{i,m},
	\label{eq:ce_loss}
\end{equation}
\noindent where $\mathbb{I}\left(\cdot\right)$ is the indicator function, and $\hat{p}_{i,m}$ is the softmax-normalized probability, i.e.
\begin{equation}
	\hat{p}_{i,m} = \frac{e^{s_{i,m}}}{\sum_{\ell=1}^{M} e^{s_{i,\ell}}}.
	\label{eq:softmax}
\end{equation}

To solve \eqref{eq:optimization} via stochastic gradient descent, the trainable
parameters are first partitioned into two disjoint subsets, i.e.
\begin{align}
	\boldsymbol{\theta} = \boldsymbol{\theta}_{\text{base}} \cup \boldsymbol{\theta}_{\text{film}}, \label{eq:param_union} \\
	\boldsymbol{\theta}_{\text{base}} \cap \boldsymbol{\theta}_{\text{film}} = \varnothing, \label{eq:param_disjoint}
\end{align}
\noindent where
$\boldsymbol{\theta}_{\text{film}} = \left\{\bm{W}_\gamma, \bm{b}_\gamma,
\bm{W}_\beta, \bm{b}_\beta\right\}$ denotes the parameters of the FiLM
correction module defined in \eqref{eq:film_linear_gamma} and
\eqref{eq:film_linear_beta}, and $\boldsymbol{\theta}_{\text{base}}$ collects all
remaining parameters from the dual-branch feature extraction module and the
classifier module.  Let $\eta\left(\ell\right)$ denote the learning rate at the $\ell$-th gradient update. Then the parameter updates are performed with differentiated learning rates as
\begin{align}
	\boldsymbol{\theta}_{\text{base}}^{\left(e+1\right)} &=
	\boldsymbol{\theta}_{\text{base}}^{\left(e\right)} - \eta\left(e\right) \,
	\nabla_{\boldsymbol{\theta}_{\text{base}}} \mathcal{L}_{\text{CE}},
	\label{eq:update_base} \\
	\boldsymbol{\theta}_{\text{film}}^{\left(e+1\right)} &=
	\boldsymbol{\theta}_{\text{film}}^{\left(e\right)} - \alpha \, \eta\left(e\right) \,
	\nabla_{\boldsymbol{\theta}_{\text{film}}} \mathcal{L}_{\text{CE}},
	\label{eq:update_film}
\end{align}
\noindent where $\alpha$ is the film learning rate multiplier. This asymmetric update strategy allows the FiLM correction module to adapt at a different pace from the backbone feature extractor. Thus ends the description of the proposed EAFM identifier.

\begin{remark}
	The proposed EAFM framework leverages the channel-conditioning information embedded in the normalized envelope to adaptively correct the IQ-domain device features via a dual-branch architecture. Lemma~\ref{lemma1} establishes the monotonic relationship between the envelope CV and the Rician $K$-factor, motivating the use of envelope statistics to guide the FiLM-based modulation. This mechanism enables the identifier to adjust its outputs dynamically under varying channel conditions, thereby achieving robust closed-set emitter identification across channels.
\end{remark}

\section{Proposed EAFM-MD Verifier for SEV}
\label{chapter:4}
In this section, we propose the EAFM-MD verifier for SEV, so that our EAFM identifier can detect unknown spoofing transmitters without attack samples during training, as depicted in Fig.~\ref{fig:sev_method}. Firstly, the core innovation replaces the softmax classifier with a fingerprint library in the FiLM-corrected feature space, where each legitimate device is represented by its statistical distribution. This subsequently enables verification of unknown spoofers using the Mahalanobis distance. Finally, we optimize the system using a metric-learning objective that learns compact, separable feature representations, along with adaptive threshold selection that balances identification accuracy and spoofing-detection performance.

\begin{figure*}[t]
	\centering
	\includegraphics[width=0.85\textwidth]{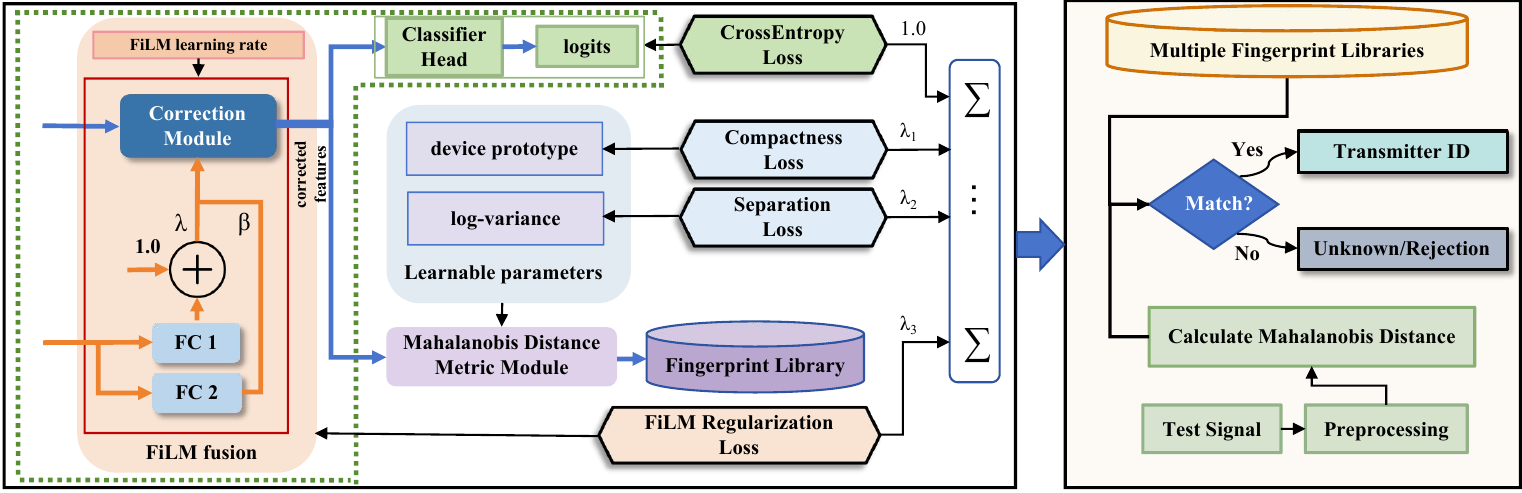}
	\caption{EAFM identifier to the SEV setting by introducing a fingerprint library that stores per-device centroids and inverse covariance matrices, as well as the within-class Mahalanobis distances of training samples for threshold determination. During deployment, the verifier computes the Mahalanobis distance between the extracted feature and each known device centroid, rejects the sample as Eve if the minimum distance exceeds the device-specific threshold, and otherwise identifies the device using the classifier logits.}
	\label{fig:sev_method}
\end{figure*}

\subsection{EAFM-MD Verifier Architecture}
\label{eafm_md_framework}

The proposed EAFM-MD verifier $\mathcal{G}_{\mathrm{o}}\left(\cdot,\cdot\right)$ is built upon the EAFM identifier $\mathcal{G}_{\mathrm{c}}\left(\cdot,\cdot\right)$ from Section~\ref{chapter:3} and extends it to implement the verification mapping defined in \eqref{eq:open_set_sei}. Specifically, we replace the softmax classifier of $\mathcal{G}_{\mathrm{c}}$ with three cascaded stages: a feature extraction stage $\mathcal{E}\left(\cdot,\cdot\right)$ that maps the input to a FiLM-corrected feature vector, a fingerprint library and matching stage that constructs the per-device statistical model from the training corpus and computes the Mahalanobis distances to known device distributions during deployment, and a verification decision stage $\mathcal{D}\left(\cdot,\cdot\right)$ that determines the predicted identity by jointly considering the distance vector and the classifier logits. The predicted identity is obtained by composing these three stages as
\begin{align}
	\widehat{m}_i^{\mathrm{o}} &= \mathcal{G}_{\mathrm{o}}\left(\bm{Y}_i,\, \hat{\bm{r}}_i\right) \label{eq:eafm_md_composition_a}, \\
	&= \mathcal{D}\left( \mathcal{M}\left( \mathcal{E}\left(\bm{Y}_i,\, \hat{\bm{r}}_i\right) \right),\, \bm{s}_i \right) \label{eq:eafm_md_composition_b},
\end{align}
\noindent where $\bm{s}_i$ denotes the classifier logits inherited from the EAFM identifier. In what follows, we elaborate on the three stages, beginning with the feature extraction stage, followed by the fingerprint library and matching stage, and concluding with the verification decision stage.

\paragraph{Feature extraction stage}
The feature extraction stage inherits the architecture of the EAFM identifier. The FiLM-corrected feature vector is obtained by
\begin{align}
	\mathcal{E}\left(\bm{Y}_i,\, \hat{\bm{r}}_i\right) &= \bm{f}_{\text{corr},i} \in \mathbb{R}^{d}, \label{eq:feature_extraction} \\
	&= \mathcal{F}\left( \varPhi_{\text{main}}\left(\bm{Y}_i\right),\, \varPhi_{\text{aux}}\left(\hat{\bm{r}}_i\right) \right) \label{eq:feature_extraction_film},
\end{align}
\noindent where $\varPhi_{\text{main}}\left(\cdot\right)$, $\varPhi_{\text{aux}}\left(\cdot\right)$, and $\mathcal{F}\left(\cdot,\cdot\right)$ are the main branch extractor, auxiliary branch extractor, and FiLM correction module defined in \eqref{eq:main_feature}, \eqref{eq:aux_feature}, and \eqref{eq:film_operator}, respectively. This stage is identical to the feature extraction component of the EAFM identifier in \eqref{eq:eafm_composition_b}, ensuring that all representational capabilities developed in the identification setting are retained.

\paragraph{Fingerprint library and matching stage}
The fingerprint library and the matching stage constitute the core component that enables the verifier to detect unknown transmitters without requiring any attack signals during training. The matching stage $\mathcal{M}\left(\cdot\right)$ computes a distance vector $\bm{d}_i$ by evaluating the Mahalanobis distance between the extracted feature $\bm{f}_{\text{corr},i}$ obtained from $\mathcal{E}\left(\cdot,\cdot\right)$ and each known device distribution using the library parameters, thereby quantifying the statistical conformity of the input to each legitimate device. Given the extracted feature $\bm{f}_{\text{corr},i}$, the matching stage is formally defined as
\begin{align}
	\mathcal{M}\bigl(\bm{f}_{\text{corr},i}\bigr) &= \bm{d}_i \in \mathbb{R}^{M}, \label{eq:matching_stage} \\
	\left[\bm{d}_i\right]_m &= \mathsf{d}_M\!\bigl(\bm{f}_{\text{corr},i};\, \boldsymbol{\mu}_m,\, \boldsymbol{\Sigma}_m^{-1}\bigr), \quad m \in \mathcal{I}_M, \label{eq:matching_element}
\end{align}
\noindent where $d_M\!\left(\cdot;\, \boldsymbol{\mu}_m,\, \boldsymbol{\Sigma}_m^{-1}\right)$ is the Mahalanobis distance, i.e.
\begin{equation}
	\ d_M\!\left(\bm{f};\, \boldsymbol{\mu}_m,\, \boldsymbol{\Sigma}_m^{-1}\right) = \sqrt{ \left( \bm{f} - \boldsymbol{\mu}_m \right)^{\mathsf{T}} \boldsymbol{\Sigma}_m^{-1} \left( \bm{f} - \boldsymbol{\mu}_m \right) },
	\label{eq:mahalanobis_distance}
\end{equation}
\noindent where $\boldsymbol{\mu}_m$ and $\boldsymbol{\Sigma}_m^{-1}$ are the centroid and inverse covariance matrix of device $m$, which are stored in the fingerprint library $\mathcal{L}$. The resulting distance vector $\bm{d}_i$ is then passed to the verification decision stage for identity determination.

The fingerprint library $\mathcal{L}$ is constructed from the training corpus using the trained EAFM backbone. It stores per-device statistical parameters that characterize the distribution of legitimate samples in the FiLM-corrected feature space, thereby enabling the matching stage $\mathcal{M}\left(\cdot\right)$ to compute the distance vector during deployment. Specifically, let $\mathcal{S}_m = \left\{i \mid m_i^{\mathrm{c}} = m\right\}$ denote the index set of training samples belonging to device $m$, and let $S_m = |\mathcal{S}_m|$ be the number of such samples. The fingerprint library is defined as a collection of triplets, one per known device, which can be given by
\begin{equation}
	\mathcal{L} = \left\{ \bigl( \boldsymbol{\mu}_m,\, \boldsymbol{\Sigma}_m^{-1},\, \mathcal{W}_m \bigr) \right\}_{m=1}^M,
	\label{eq:fingerprint_library_definition}
\end{equation}
\noindent where $\boldsymbol{\mu}_m$ is the centroid of device $m$, $\boldsymbol{\Sigma}_m^{-1}$ is the inverse covariance matrix used for the Mahalanobis distance computation at deployment, and $\mathcal{W}_m$ denotes the set of within-class Mahalanobis distances of the training samples of device $m$ for threshold determination.

The centroid is estimated as the sample mean of the FiLM-corrected features of device $m$, i.e.
\begin{equation}
	\boldsymbol{\mu}_m = \frac{1}{S_m} \sum_{i \in \mathcal{S}_m} \bm{f}_{\text{corr},i}.
	\label{eq:centroid_estimation}
\end{equation}

For the full covariance matrix, we adopt the sample covariance with a regularization term that ensures positive definiteness, which can be formulated as
\begin{align}
	\boldsymbol{\Sigma}_m &= \frac{1}{S_m - 1} \sum_{i \in \mathcal{S}_m} \left(\bm{f}_{\text{corr},i} - \boldsymbol{\mu}_m\right)\left(\bm{f}_{\text{corr},i} - \boldsymbol{\mu}_m\right)^{\mathsf{T}} \nonumber \\
	&\quad + \epsilon \bm{I}_d,
	\label{eq:covariance_estimation}
\end{align}
\noindent where $\epsilon > 0$ is a small regularization constant and $\bm{I}_d$ is the identity matrix. Its inverse $\boldsymbol{\Sigma}_m^{-1}$ is then computed and stored in the fingerprint library.

For each training sample belonging to device $m$, we compute the within-class Mahalanobis distance to its own centroid using \eqref{eq:mahalanobis_distance}, and the resulting set of distances is stored in the fingerprint library as the statistical basis for threshold determination, i.e.
\begin{equation}
	\mathcal{W}_m = \left\{ d_M\!\left(\bm{f}_{\text{corr},i};\, \boldsymbol{\mu}_m,\, \boldsymbol{\Sigma}_m^{-1}\right) \right\}_{i \in \mathcal{S}_m}.
	\label{eq:within_class_distance_set}
\end{equation}

We then derive the decision threshold for device $m$ from the empirical distribution of its within-class distances. Specifically, the threshold is set as the $p$-th percentile of $\mathcal{W}_m$, which can be expressed as
\begin{equation}
	\tau_m = \left( F_m \right)^{-1}\left( p \right),
	\label{eq:threshold_definition}
\end{equation}
\noindent where $p \in (0,1)$ is a configurable percentile parameter and $F_m(\cdot)$ denotes the empirical cumulative distribution function of $\mathcal{W}_m$ for device $m$. This formulation decouples the fingerprint library data $\mathcal{W}_m$ from the thresholding policy governed by $p$: the threshold can be adjusted at test time by recomputing the percentile on the stored within-class distances without reconstructing the entire fingerprint library.

\paragraph{Verification decision stage}
The verification decision stage $\mathcal{D}\left(\cdot,\cdot\right)$ determines the predicted identity via a two-stage rule that first rejects unknown transmitters by comparing the minimum Mahalanobis distance against the device-specific threshold, and then identifies known devices using the classifier logits. Given the distance vector $\bm{d}_i$ and logits $\bm{s}_i$, this decision can be denoted by
\begin{align}
	\widehat{m}_i^{\mathrm{o}} &= \mathcal{D}\left(\bm{d}_i,\, \bm{s}_i\right), \label{eq:open_set_decision_a} \\
	&= 
	\begin{cases}
		\displaystyle \operatorname*{arg\,max}_{m \in \mathcal{I}_M} \left[\bm{s}_i\right]_m, & \text{if } \min\limits_{m \in \mathcal{I}_M} \left[\bm{d}_i\right]_m \leq \tau_m, \\[6pt]
		\text{Eve}, & \text{otherwise},
	\end{cases} \label{eq:open_set_decision_b}
\end{align}
\noindent where $\tau_m$ is the device-specific threshold derived from the fingerprint library via \eqref{eq:threshold_definition}. The rationale of this two-stage design is that the Mahalanobis distance in the FiLM-corrected feature space captures the statistical conformity of the input to the known device distribution and thus serves as a reliable indicator for detecting unknown transmitters, while the classifier logits retain the fine-grained discriminability required for multi-device identification among legitimate devices.

\subsection{Optimization of EAFM-MD Verifier}
\label{eafm_md_optimization}

Unlike the identification setting in \eqref{eq:optimization}, the training set does not contain any attack signals. The overall objective is a weighted combination of four loss terms, which can be expressed as
\begin{equation}
	\begin{aligned}
		\boldsymbol{\psi}^{*} &= \operatorname*{arg\,min}_{\boldsymbol{\psi}} \; 
		\frac{1}{S} \sum_{i \in \mathcal{I}_S} \Big[ \mathcal{L}_{\text{CE}}\left( \bm{s}_i,\, m_i^{\mathrm{c}} \right) \\
		&\quad + \lambda_1 \mathcal{L}_{\text{compact}}^{(i)} + \lambda_2 \mathcal{L}_{\text{sep}} + \lambda_3 \mathcal{L}_{\text{film}}^{(i)} \Big],
	\end{aligned}
	\label{eq:eafm_md_optimization}
\end{equation}
\noindent where $\boldsymbol{\psi}^{*}$ and $\boldsymbol{\psi}$ denote the optimal parameter and the trainable parameters, respectively; $\mathcal{L}_{\text{compact}}^{(i)}$ is the compactness loss~\cite{refWen} that penalizes feature-to-centroid distances; $\mathcal{L}_{\text{sep}}$ is the separation loss~\cite{refNguyen} that enforces inter-centroid margins; $\mathcal{L}_{\text{film}}^{(i)}$ is the FiLM regularization loss that stabilizes the modulation parameters; and $\lambda_1$, $\lambda_2$, $\lambda_3$ are the corresponding weighting hyperparameters.\footnote{The compactness and separation losses follow standard metric learning formulations, while the innovation of EAFM-MD lies in its integration with the adaptive FiLM backbone and Mahalanobis fingerprint library, rather than in the loss design itself.}

\paragraph{Compactness loss}
The compactness loss penalizes the diagonal Mahalanobis distance between each training sample's feature and its corresponding device centroid, which is defined as
\begin{equation}
	\mathcal{L}_{\text{compact}}^{(i)} = \sqrt{ \left( \bm{f}_{\text{corr},i} - \boldsymbol{\mu}_{m_i^{\mathrm{c}}} \right)^{\mathsf{T}} \boldsymbol{\Lambda}_{m_i^{\mathrm{c}}}^{-1} \left( \bm{f}_{\text{corr},i} - \boldsymbol{\mu}_{m_i^{\mathrm{c}}} \right) },
	\label{eq:compact_loss}
\end{equation}
\noindent where $\boldsymbol{\mu}_{m_i^{\mathrm{c}}}$ is the learnable centroid of the device indexed by the label $m_i^{\mathrm{c}}$, and $\boldsymbol{\Lambda}_{m}^{-1} = \operatorname{diag}\bigl(\sigma_{m,1}^{-2}, \ldots, \sigma_{m,d}^{-2}\bigr)$ is the learnable diagonal precision matrix for device $m$, which is parameterized through $\log \boldsymbol{\sigma}_m^2$ to adaptively identify reliable feature dimensions for each device.

Notably, the diagonal precision approximation reduces the number of learnable parameters from $\mathcal{O}\left(d^2\right)$ to $\mathcal{O}\left(d\right)$ per device and avoids inverting a full covariance matrix during training. However, it discards inter-dimensional correlations, which we address during deployment by constructing the fingerprint library with full covariance matrices estimated from the training corpus, as described in Section~\ref{eafm_md_framework}.

\paragraph{Separation loss}
To enhance inter-class discriminability and prevent centroid collapse, the separation loss adopts a hinge-style formulation that penalizes pairs of centroids that are insufficiently separated, i.e.
\begin{equation}
	\mathcal{L}_{\text{sep}} = \frac{1}{N_{\text{pair}}} \sum_{m=1}^{M} \sum_{m' > m} \max\!\left(0,\, \delta - \|\boldsymbol{\mu}_m - \boldsymbol{\mu}_{m'}\|_2 \right),
	\label{eq:sep_loss}
\end{equation}
\noindent where $N_{\text{pair}}  = M\left(M-1\right)/2$ is the number of device pairs, $\boldsymbol{\mu}_m$ and $\boldsymbol{\mu}_{m'}$ are learnable centroids of devices $m$ and $m'$, and $\delta$ is a positive margin hyperparameter. Unlike a simple negative Euclidean distance penalty, the hinge formulation only activates when the inter-centroid distance falls below $\delta$, thereby focusing the gradient on poorly separated pairs while allowing well-separated centroids to remain unconstrained.

\paragraph{FiLM regularization loss}
We stabilize the FiLM correction mechanism during joint training with the fingerprint parameters by applying a regularization loss that penalizes deviations of the FiLM modulation parameters from their identity-preserving initialization, i.e.
\begin{equation}
	\mathcal{L}_{\text{film}}^{(i)} = \|\boldsymbol{\gamma}_i - \bm{1}\|_2^2 + \|\boldsymbol{\beta}_i\|_2^2,
	\label{eq:film_reg_loss}
\end{equation}
\noindent where $\boldsymbol{\gamma}_i$ and $\boldsymbol{\beta}_i$ are the FiLM modulation parameters defined in \eqref{eq:film_output}. This loss encourages the FiLM modulation to remain close to the identity operation unless substantial evidence from the envelope branch suggests otherwise.

We address \eqref{eq:eafm_md_optimization} with stochastic gradient descent by first partitioning the trainable parameters into three disjoint subsets, which is written as
\begin{align}
	\boldsymbol{\psi} &= \boldsymbol{\psi}_{\text{base}} \cup \boldsymbol{\psi}_{\text{film}} \cup \boldsymbol{\psi}_{\text{fingerprint}}, \label{eq:eafm_md_param_partition} \\
	\boldsymbol{\psi}_{\text{fingerprint}} &= \left\{\boldsymbol{\mu}_1, \ldots, \boldsymbol{\mu}_M,\, \log \boldsymbol{\sigma}_1^2, \ldots, \log \boldsymbol{\sigma}_M^2\right\}, \label{eq:eafm_md_fingerprint_params}
\end{align}
\noindent where $\boldsymbol{\psi}_{\text{base}}$ collects the parameters of the dual-branch feature extraction module and the classifier module, $\boldsymbol{\psi}_{\text{film}}$ collects the parameters of the FiLM correction module defined in \eqref{eq:film_linear_gamma} and \eqref{eq:film_linear_beta}, and $\boldsymbol{\psi}_{\text{fingerprint}}$ collects the learnable centroid and log-variance parameters. The parameter updates follow the differentiated learning rate strategy established in \eqref{eq:update_base} and \eqref{eq:update_film}, with $\boldsymbol{\psi}_{\text{fingerprint}}$ updated at the base learning rate alongside $\boldsymbol{\psi}_{\text{base}}$ and $\boldsymbol{\psi}_{\text{film}}$ updated at a modulated rate $\alpha\,\eta\left(e\right)$.

\subsection{Deployment and Optimal $\alpha^{*}$ Selection}
\label{deployment}

After training, the deployment performance is critically influenced by the hyperparameter $\alpha$, which governs the trade-off between FiLM adaptability and feature stability in the learned fingerprint space.
Because no single value of $\alpha$ is universally optimal across different channel conditions and device sets, we train $L$ independent EAFM-MD verifiers over a candidate set $\{\alpha_1, \alpha_2, \ldots, \alpha_L\}$, each yielding a distinct set of backbone parameters $\boldsymbol{\psi}_{\alpha_\ell}$ and a corresponding fingerprint library $\mathcal{L}^{(\alpha_\ell)}$. The optimal candidate $\alpha^{*}$ is then selected on a held-out validation set by balancing identification accuracy against spoofing detection performance, as detailed below.

Since both the identification accuracy on legitimate devices and the detection rate of unknown transmitters are equally important for SEV, we select the optimal $\alpha^{*}$ by maximizing their equally-weighted harmonic mean, i.e.
\begin{equation}
	\alpha^{*} = \operatorname*{arg\,max}_{\alpha_\ell \in \{\alpha_1,\ldots,\alpha_L\}} \;
	\frac{2\,\mathrm{Acc}^{(\alpha_\ell)} \cdot \mathrm{PD}^{(\alpha_\ell)}}{\mathrm{Acc}^{(\alpha_\ell)} + \mathrm{PD}^{(\alpha_\ell)}},
	\label{eq:best_alpha_selection}
\end{equation}
\noindent where $\mathrm{Acc}^{(\alpha_\ell)}$ and $\mathrm{PD}^{(\alpha_\ell)}$ denote the identification accuracy on legitimate devices and the detection rate of unknown transmitters, respectively, which are formulated as
\begin{align}
	\mathrm{Acc}^{(\alpha_\ell)} &= 
	\frac{\bigl|\bigl\{ i : \widehat{m}_i^{\mathrm{o}}{}^{(\alpha_\ell)} = m_i^{\mathrm{c}} \mid m_i^{\mathrm{c}} \in \mathcal{I}_M \bigr\}\bigr|}
	{\bigl|\bigl\{ i : m_i^{\mathrm{c}} \in \mathcal{I}_M \bigr\}\bigr|}, \label{eq:identification_accuracy} \\[4pt]
	\mathrm{PD}^{(\alpha_\ell)} &= 
	\frac{\bigl|\bigl\{ i : \widehat{m}_i^{\mathrm{o}}{}^{(\alpha_\ell)} = \text{Eve} \mid m_i^{\mathrm{c}} = \text{Eve} \bigr\}\bigr|}
	{\bigl|\bigl\{ i : m_i^{\mathrm{c}} = \text{Eve} \bigr\}\bigr|}. \label{eq:detection_rate}
\end{align}

The optimal $\alpha^{*}$ selection procedure is summarized in Algorithm~\ref{alg:alpha_selection}.

\begin{algorithm}[t]
	\caption{Optimal $\alpha^{*}$ Selection Algorithm}
	\label{alg:alpha_selection}
	\begin{algorithmic}
		\REQUIRE $\{\boldsymbol{\psi}_{\alpha_\ell},\, \mathcal{L}^{(\alpha_\ell)}\}_{\ell=1}^{L}$, $\{(\bm{Y}_i, \hat{\bm{r}}_i, m_i^{\mathrm{c}})\}$, and $p$;
		\ENSURE Optimal $\alpha^{*}$.
		\FOR{$\ell = 1$ to $L$}
		\STATE Load $\boldsymbol{\psi}_{\alpha_\ell}$ and $\mathcal{L}^{(\alpha_\ell)}$;
		\STATE Extract $\{\bm{f}_{\text{corr},i}^{(\alpha_\ell)}\}$ using $\boldsymbol{\psi}_{\alpha_\ell}$;
		\STATE Compute $\{\bm{d}_i^{(\alpha_\ell)}\}$ via \eqref{eq:matching_stage}--\eqref{eq:matching_element};
		\STATE Obtain $\{\widehat{m}_i^{\mathrm{o}}{}^{(\alpha_\ell)}\}$ via \eqref{eq:open_set_decision_b};
		\STATE Evaluate $\mathrm{Acc}^{(\alpha_\ell)}$ and $\mathrm{PD}^{(\alpha_\ell)}$ via \eqref{eq:identification_accuracy}--\eqref{eq:detection_rate};
		\ENDFOR
		\STATE Select $\alpha^{*}$ via \eqref{eq:best_alpha_selection};
		\RETURN $\alpha^{*}$.
	\end{algorithmic}
\end{algorithm}

\begin{remark}
	\label{remark:eafm_vs_eafmmd}
	The proposed EAFM-MD verifier extends the EAFM identifier to SEV without altering the dual-branch architecture or the FiLM correction mechanism. Unlike the EAFM identifier, which relies on a softmax classifier that inevitably assigns every input to a known device, the EAFM-MD verifier replaces the softmax with a fingerprint library that stores per-device centroids, inverse covariance matrices, and within-class Mahalanobis distances, and makes the verification decision by comparing the minimum Mahalanobis distance against device-specific thresholds. This design enables the detection of unknown spoofing transmitters without requiring any attack samples during training.
\end{remark}

\setcounter{table}{1}
\begin{table*}[t]
	\centering
	\caption{\textsc{Transmitter parameters}}
	\label{tab:device-params}
	\begin{tabular}{lcccccc}
		\toprule
		Device & \(G\) & \phantom{-}\(\zeta\) \(\cdot\frac{\pi}{180}\) & \( a^{\mathrm{ST}} \) & \( f^{\mathrm{ST}} \) \textmd{(MHz)} & \(\xi_m \times 10^{-3}\) & [$b_1$,$b_2$,$b_3$] \\
		\midrule
		Dev1 & 0.9998 & -0.0180\ & 0.0082 & 0.129 & 1.3 + 8.2\(j\)  & [1.00, 0.50, 0.30] \\
		Dev2 & 1.0056 & \phantom{-}0.0175\ & 0.0075 & 0.132 & 1.5 + 7.2\(j\) & [1.00, 0.08, 0.60] \\
		Dev3 & 1.0102 & \phantom{-}0.0120\ & 0.0070 & 0.123 & 1.1 + 6.8\(j\) & [1.00, 0.01, 0.01] \\
		Dev4 & 0.9992 & \phantom{-}0.0030\ & 0.0087 & 0.135 & 1.7 + 9.0\(j\) & [1.00, 0.01, 0.40] \\
		Dev5 & 0.9500 & \phantom{-}0.0300\ & 0.0195 & 0.165 & 3.2 + 13.5\(j\) & [1.00, 0.95, 0.80] \\
		\bottomrule
	\end{tabular}
\end{table*}

\setcounter{table}{0}
\begin{table}[t]
	\centering
	\caption{System Simulation Parameters}
	\label{tab:system-params}
	\begin{tabular}{lc}
		\toprule
		\textbf{Parameter} & \textbf{Value} \\
		\midrule
		QPSK symbol number $N_{\text{symbols}}$ & 512 \\
		Samples per symbol $\text{SPS}$ & 4 \\
		Root-raised cosine rolloff $\alpha$ & 0.3 \\
		RRC filter span $M_{\text{span}}$ & 8 \\
		Sampling rate $F_s$ & 30.72 MHz \\
		Carrier frequency $F_c$ & 2.4 GHz \\
		Transmitter antenna height $h_{\text{BS}}$ & 10 m \\
		Receiver antenna height $h_{\text{UT}}$ & 100 m \\
		Two-dimensional distance $d_{2\text{D}}$ & 500 m \\
		Three-dimensional distance $d_{3\text{D}}$ & 508.0 m \\
		Rician $K$-factors & $\left\{-2,2,6,10\right\}$ dB \\
		\bottomrule
	\end{tabular}
\end{table}

\setcounter{table}{2}
\begin{table}[t]
	\centering
	\caption{Training Configuration}
	\label{tab:training-loss-hyperparams}
	\begin{tabular}{lc}
		\toprule
		\textbf{Parameter} & \textbf{Default Value} \\
		\midrule
		Total training epochs & 200 \\
		Batch size & 256 \\
		Initial learning rate & 5e-4 \\
		Minimum learning rate & 1e-5 \\
		Weight decay (L2 regularization) & 5e-4 \\
		Warmup epochs & 10 \\
		$\left\{\lambda_1, \lambda_2, \lambda_3\right\}$ & $\left\{0.2, 0.2, 0.2\right\}$ \\
		Gradient clipping max norm & 1.0 \\
		Target distance for class separation & 5.0 \\
		Covariance regularization $\epsilon$ & $10^{-6}$ \\
		\bottomrule
	\end{tabular}
\end{table}

\section{Numerical Results}
In this section, we evaluate the performance of the proposed EAFM identifier for SEI and the proposed EAFM-MD verifier for SEV under cross-Rician $K$-factor OOD settings and unknown spoofing attacks. We first describe the simulation setup and baselines, and then present numerical results for three scenarios. The first scenario examines identification accuracy across mismatched $K$ conditions. The second scenario evaluates verification with spoofing attack devices. The third scenario investigates identification and verification with both cross-channel OOD and unknown-device detection.

\subsection{Simulation setup and baselines}
\label{sec:sim-setup}

We use a synthetic dataset generated from the Rician fading channel model for model training and evaluation. The dataset considers RF hardware impairments and multipath propagation. The system parameters used in the simulation are shown in Table~\ref{tab:system-params}. Referencing the 3GPP UMa (Urban Macro) scenario, the transmitter antenna height is set to $h_{\text{BS}}=10\text{ m}$, receiver antenna height $h_{\text{UT}}=100\text{ m}$, and two-dimensional distance $d_{2\text{D}}=500\text{ m}$. The three-dimensional distance is calculated as
\begin{equation}
	d_{3\text{D}} = \sqrt{d_{2\text{D}}^2 + (h_{\text{UT}} - h_{\text{BS}})^2} \approx 508.0\text{ m}.
\end{equation}

The transmitter integrates multiple RF hardware impairments as described in Section \ref{chapter:system_model_and_problem}, with device parameters detailed in Table~\ref{tab:device-params}. Specifically, Dev5 represents the spoofing attack device used for testing, while the remaining devices are legitimate. Furthermore, we employ a Rician fading multipath channel model with the following 
multipath delay and power configuration: Path delays are $\{0, 80, 200, 570, 1090, 
1730, 2510\}$ ns with corresponding power levels $\{0.0, -2.7, -3.0, -4.6, -7.5, 
-10.6, -13.1\}$ dB. The early stopping patience is set to 30 for classification and 20 for detection, and the other training hyperparameters are summarized in Table~\ref{tab:training-loss-hyperparams}.

\setcounter{table}{3}
\begin{table*}[t]
	\centering
	\caption{Device Acc of Different Methods under Various Rician $K$ Conditions}
	\label{tab:rician_k_factor}
	\begin{tabular}{|c|ccc|c|ccc|c|ccc|c|}
		\toprule
		\textbf{Model} & \multicolumn{4}{c|}{\textbf{ID(4dB)}} & \multicolumn{4}{c|}{\textbf{OOD(-5dB)}} & \multicolumn{4}{c|}{\textbf{OOD(-10dB)}} \\
		\midrule
		SNR/Avg.& 0 dB & 5 dB & 10 dB & Avg. & 0 dB & 5 dB & 10 dB & Avg. & 0 dB & 5 dB & 10 dB & Avg. \\
		\midrule
		Baseline & 0.8381 & 0.9353 & 0.9681 & 0.9138 & 0.7390 & 0.8912 & 0.9536 & 0.8613 & 0.6809 & 0.8337 & 0.9227 & 0.8125 \\
		
		DANN & 0.7629 & 0.9284 & 0.9576 & 0.8830 & 0.6279 & 0.9011 & 0.9234 & 0.8175 & 0.5515 & 0.8654 & 0.8860 & 0.7676 \\
		
		VREx & 0.8196 & 0.9147 & 0.9786 & 0.9043 & 0.7411 & 0.8203 & 0.9479 & 0.8364 & 0.6907 & 0.7461 & 0.9056 & 0.7808 \\
		
		Proposed EAFM ($\alpha$=0.0) & 0.8655 & 0.9456 & 0.9781 & 0.9297 & 0.7759 & 0.9093 & 0.9160 & 0.8672 & 0.7010 & 0.8714 & 0.8552 & 0.8092 \\
		\midrule
		\textbf{Proposed EAFM} ($\alpha$=2.0) & \textbf{0.8726} & \textbf{0.9554} & \textbf{0.9885} & \textbf{0.9388} & \textbf{0.8386} & \textbf{0.9286} & \textbf{0.9768} & \textbf{0.9147} & \textbf{0.8096} & \textbf{0.9066} & \textbf{0.9551} & \textbf{0.8905} \\
		\bottomrule
	\end{tabular}
\end{table*}

With a given SNR, the training set contains 32,000 samples. Four legitimate devices generate 2,000 samples for each of four Rician $K$-factor conditions $\{-2, 2, 6, 10\}\,\mathrm{dB}$, so the total number of training samples is 32,000. This training set contains no attack samples. The validation set consists of 8,000 samples generated under the same $K$ conditions, and a $4:1$ training-to-validation split is applied so that the validation set contains no attack samples. The testing stage evaluates two scenarios under the same SNR. (i) Identification scenario: Four legitimate devices generate 2,000 samples for each of three testing $K$ conditions $\{4, -5, -10\}\,\mathrm{dB}$, so the total number of testing samples is 8,000; (ii) Spoofing attack scenario: Four legitimate devices and one malicious device (Dev5) generate 2,000 samples for each of the same three testing $K$ conditions $\{4, -5, -10\}\,\mathrm{dB}$, so the total number of testing samples is 10,000.\footnote{The dataset is publicly available at https://www.scidb.cn/s/B3QN3y} Additionally, the threshold percentile parameter in \eqref{eq:threshold_definition} is set to $p=0.95$ for scenario (ii). We conducted three types of OOD experiments as follows: 
\begin{enumerate}
	\item For the SEI scenario (i), the testing set is generated using four legitimate devices under the same SNR with $K\in\{4, -5, -10\}\,\mathrm{dB}$; note that $K=-5\,\mathrm{dB}$ and $K=-10\,\mathrm{dB}$ are OOD values not included in the training set $K\in\{-2, 2, 6, 10\}\,\mathrm{dB}$. For the proposed EAFM identifier, we compare against the following baselines. The baseline model does not incorporate either the envelope branch or the FiLM module, thereby serving as a fundamental comparison method. For further reference, DANN~\cite{Ganin2016Domain} and VREx~\cite{Krueger2021OutOfDistribution} are also implemented as OOD generalization baselines. The proposed EAFM identifier is also evaluated in two variants. The first variant uses a frozen FiLM module, with the auxiliary branch present but $\alpha$ set to 0.0. The second variant applies active FiLM modulation, where $\alpha \neq 0.0$.
	
	\item For the SEV scenario (ii) with Rician $K=4\,\mathrm{dB}$, the testing set includes four legitimate devices from the training set and one OOD device (Dev5) not encountered during training. We compare the proposed EAFM-MD verifier with the EAFM-ED verifier. Both share the same EAFM architecture but differ in their distance metric. The proposed EAFM-MD verifier employs Mahalanobis distance to explicitly model both per-device variance and inter-dimensional correlations, enabling discriminative device-level separation. Specifically, the EAFM-ED verifier is inherently limited by Euclidean distance, which treats all feature dimensions equally and cannot account for per-device covariance structure by definition~\cite{Lee2018Simple}.
	
	\item For the SEV scenario (ii) with Rician $K\in\{4, -5, -10\}\,\mathrm{dB}$, the testing set includes four legitimate devices from the training set and one OOD device (Dev5) not encountered during training. This scenario constitutes the most challenging part of our evaluation, as it simultaneously involves cross-channel OOD conditions and unknown device detection. We compare the proposed EAFM-MD verifier with the Baseline-MD verifier, which removes both the envelope auxiliary branch and the FiLM correction module while retaining the same Mahalanobis-distance-based fingerprint library architecture as the proposed EAFM-MD verifier. The fingerprint library construction and optimal $\alpha^*$ selection for the EAFM-MD verifier follow Algorithm~\ref{alg:alpha_selection}. Note that when $\alpha = 0.0$, the auxiliary branch and FiLM module are effectively disabled, making the fingerprint library of the EAFM-MD verifier coincident with that of the Baseline-MD verifier.
\end{enumerate}

Here, several evaluation indicators are considered. The Probability of False Alarm (PFA) quantifies the proportion of legitimate device samples (Dev1–Dev4) that are incorrectly classified as impostors. Identification Accuracy (Acc) refers to the proportion of received signals during identity authentication in which the predicted identity matches the actual legitimate transmitter. Overall Accuracy (Overall Acc) is the proportion of all received signals for which the predicted device category matches the true category. Finally, the Probability of Detection (PD) represents the proportion of samples correctly identified as impostors when the received signal truly originates from an impostor (Illegal).

\begin{figure}[t]
	\centering
	\includegraphics[width=0.825\linewidth]{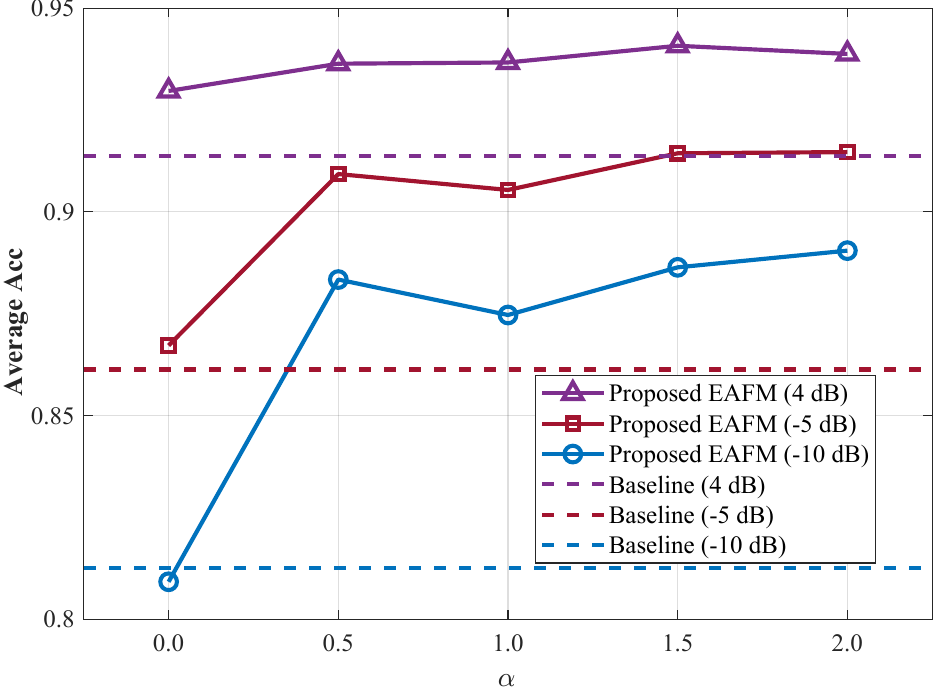}
	\caption{Average Acc of the baseline scheme and the proposed EAFM identifier under $K=4$, $-5$, and $-10\,\mathrm{dB}$.}
	\label{fig:EAFM_mean}
\end{figure}

\begin{figure}[t]
	\centering
	\includegraphics[width=0.825\linewidth]{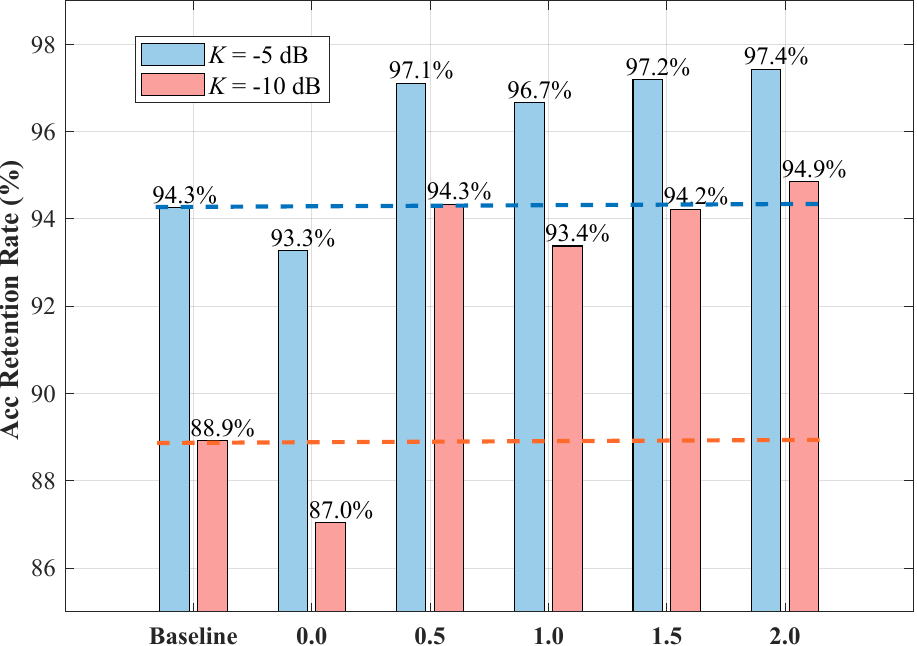}
	\caption{Retention rate of Acc for the baseline and proposed EAFM identifier with different $\alpha$ when $K=-5$ and $K=-10\,\mathrm{dB}$, relative to $K=4\,\mathrm{dB}$.}
	\label{fig:EAFM_keep}
\end{figure}

\subsection{Numerical results}
\begin{observation}
	For SEI, the proposed EAFM identifier achieves superior identification accuracy and robustness across varying Rician $K$ conditions compared to the baseline, DANN, and VREx methods. (cf. Table~\ref{tab:rician_k_factor} and Figs.~\ref{fig:EAFM_mean} and \ref{fig:EAFM_keep})
\end{observation}

Table~\ref{tab:rician_k_factor} compares the proposed EAFM identifier with the baseline, DANN, and VREx methods across varying Rician $K$ conditions. Specifically, the EAFM identifier achieves the highest average Acc in all scenarios. For instance, at $K=4\,\mathrm{dB}$, the EAFM identifier attains 0.9388, which outperforms the baseline (0.9138), DANN (0.8830), and VREx (0.9043). Under challenging conditions, the performance gains are even more pronounced. For example, when $K=-5\,\mathrm{dB}$, the EAFM identifier achieves 0.9147 versus the baseline's 0.8613. When $K=-10\,\mathrm{dB}$, the EAFM identifier attains 0.8905 versus the baseline's 0.825.

Fig.~\ref{fig:EAFM_mean} shows the average Acc with respect to $\alpha$. The proposed EAFM identifier consistently outperforms the baseline. At $K=4\,\mathrm{dB}$, the EAFM identifier ($\alpha=0.5$) reaches approximately 0.936 accuracy and surpasses the baseline's 0.914; at $K=-10\,\mathrm{dB}$, the EAFM identifier achieves around 0.883 versus the baseline's 0.813.

Furthermore, Fig.~\ref{fig:EAFM_keep} presents the retention rate of Acc relative to the ideal condition $K=4\,\mathrm{dB}$. The EAFM identifier generally maintains a higher retention rate, which indicates its robustness to channel degradation. At a FiLM LR of 0.5, the EAFM identifier retains 97.1\% and 94.3\% of its accuracy at $K=-5\,\mathrm{dB}$ and $K=-10\,\mathrm{dB}$, respectively. 

The results show that the envelope-guided FiLM modulation adapts the main-branch features to the auxiliary-branch envelope. This mechanism preserves device-specific fingerprints that are otherwise obscured by channel variations, thereby accounting for the improved identification accuracy and robustness of the proposed EAFM identifier.

\setcounter{table}{4}

\begin{table*}[t]
	\centering
	\caption{Attack Device Detection and Identification Performance under Various Rician $K$ Conditions}
	\label{tab:rician_k_factor_pla}
	\begin{tabular}{|c|c|ccc|ccc|ccc|}
		\toprule
		\multirow{2}{*}{\textbf{Model}} & \multirow{2}{*}{\textbf{Metric}}
		& \multicolumn{3}{c|}{\textbf{ID(4dB)}} 
		& \multicolumn{3}{c|}{\textbf{OOD(-5dB)}} 
		& \multicolumn{3}{c|}{\textbf{OOD(-10dB)}} \\
		\cmidrule(lr){3-5} \cmidrule(lr){6-8} \cmidrule(lr){9-11}
		& & 0 dB & 5 dB & 10 dB & 0 dB & 5 dB & 10 dB & 0 dB & 5 dB & 10 dB \\
		\midrule
		Baseline-MD & PD & 0.4425 & 0.9220 & 0.9395 & 0.5685 & 0.9330 & 0.9520 & 0.6340 & 0.9210 & 0.9285 \\
		& Acc & 0.8190 & 0.9001 & 0.9340 & 0.7506 & 0.8214 & 0.8760 & 0.6955 & 0.7511 & 0.8033 \\
		\midrule
		\textbf{EAFM-MD} ($\alpha$=$\alpha^{*}$) & PD & \textbf{0.8120} & \textbf{0.9790} & \textbf{0.9910} & \textbf{0.6920} & \textbf{0.9445} & \textbf{0.9775} & \textbf{0.7050} & \textbf{0.9375} & \textbf{0.9610} \\
		& Acc & \textbf{0.8194} & \textbf{0.9171} & \textbf{0.9389} & \textbf{0.7839} & \textbf{0.8865} & \textbf{0.9020} & \textbf{0.7183} & \textbf{0.8376} & \textbf{0.8495} \\
		\bottomrule
	\end{tabular}
\end{table*}

\begin{figure}[t]
	\centering
	\begin{minipage}[b]{0.241\textwidth}
		\centering
		\includegraphics[width=\textwidth]{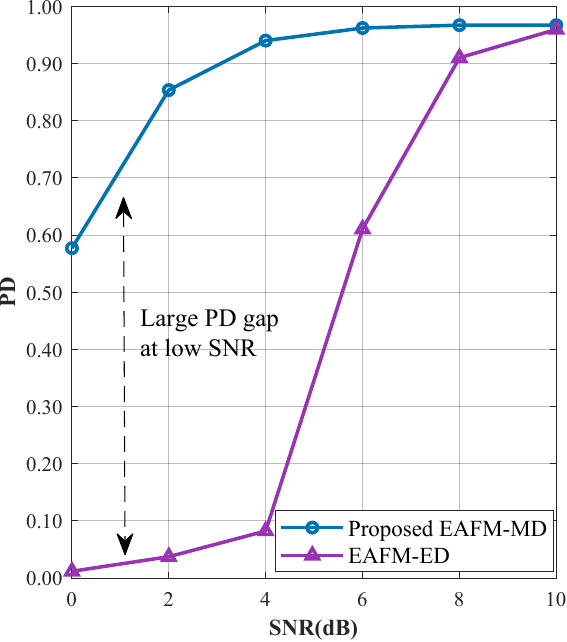}
		\par\vspace{0.5\baselineskip}
		\small (a) PD
	\end{minipage}
	\hfill
	\begin{minipage}[b]{0.241\textwidth}
		\centering
		\includegraphics[width=\textwidth]{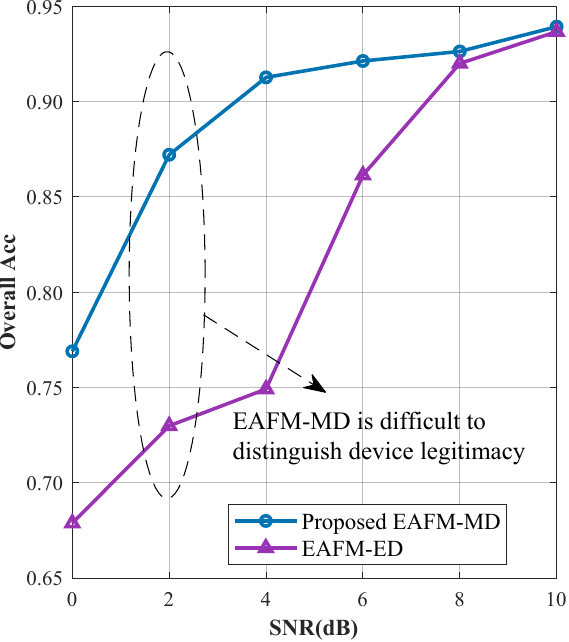}
		\par\vspace{0.5\baselineskip}
		\small (b) Overall Acc
	\end{minipage}
	\caption{PD and Overall Acc comparison for the proposed EAFM-MD verifier and the EAFM-ED verifier.}
	\label{fig:performance_metrics_comparison}
\end{figure}

\begin{figure}[t]
	\centering
	\includegraphics[width=0.485\textwidth]{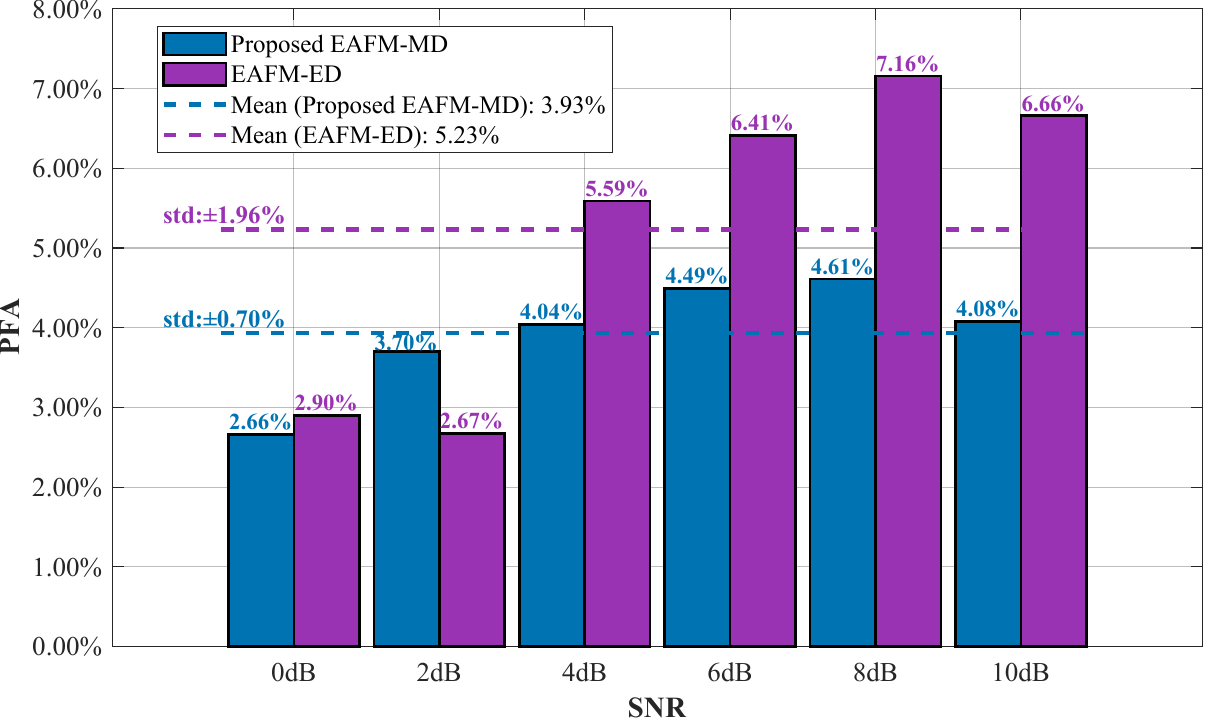}
	\caption{PFA comparison for the proposed EAFM-MD verifier and the EAFM-ED verifier. 
		Mean and standard deviation are shown for the proposed EAFM-MD verifier and the EAFM-ED verifier. 
		The proposed EAFM-MD verifier achieves the lower mean PFA and variance.}
	\label{fig:false_alarm_rate}
\end{figure}

\begin{observation}
	For SEV, the proposed EAFM-MD verifier achieves superior verification performance compared to the EAFM-ED verifier. 
	(cf. Figs.~\ref{fig:performance_metrics_comparison} and \ref{fig:false_alarm_rate})
\end{observation}

The performance comparison in Fig.~\ref{fig:performance_metrics_comparison} shows the PD and Overall Acc across SNR levels. The proposed EAFM-MD verifier consistently outperforms the EAFM-ED verifier in both metrics, with particularly large gaps at low SNR. At 0~dB, the proposed EAFM-MD verifier achieves a PD of $57.70\%$, compared to $1.20\%$ for the EAFM-ED verifier. Similarly, the EAFM-MD verifier's Overall Acc at 0~dB reaches $76.90\%$, compared with $67.89\%$ for the baseline. As the SNR increases, the performance gap narrows, but the EAFM-MD verifier still maintains an advantage. At 10~dB, the EAFM-MD verifier achieves a PD of $96.70\%$ and an Overall Acc of $93.93\%$, while the EAFM-ED verifier attains $95.95\%$ and $93.67\%$, respectively. These results indicate that the EAFM-MD verifier more reliably identifies unknown samples even under severely degraded channel conditions.

As illustrated in Fig.~\ref{fig:false_alarm_rate}, the PFA results with mean and standard deviation across all SNR levels show that the proposed EAFM-MD verifier achieves a mean PFA of $3.93\% \pm 0.64\%$, which is lower than the mean PFA of $5.23\% \pm 1.96\%$ for the EAFM-ED verifier. In addition, the EAFM-MD verifier exhibits substantially smaller variance than the EAFM-ED verifier, indicating more consistent and stable performance across noise variations.

The results indicate that the Mahalanobis distance leverages per-device covariance statistics to weight the feature dimensions, thereby capturing inter-feature correlations and per-device heterogeneity. In contrast, the Euclidean distance treats all feature dimensions equally. This enables the EAFM-MD verifier to distinguish legitimate devices from unknown devices more reliably than the EAFM-ED verifier.

\begin{observation}
	For SEV, the proposed EAFM-MD verifier consistently outperforms the Baseline-MD verifier in both the PD of the spoofing device and Acc across all Rician $K$ conditions. (cf. Table~\ref{tab:rician_k_factor_pla} and
	Figs.~\ref{fig:OOD_Detection_Performance} and~\ref{fig:Confusion_Matrix})
\end{observation}

\begin{figure}[t]
	\centering
	\includegraphics[width=0.244\textwidth]{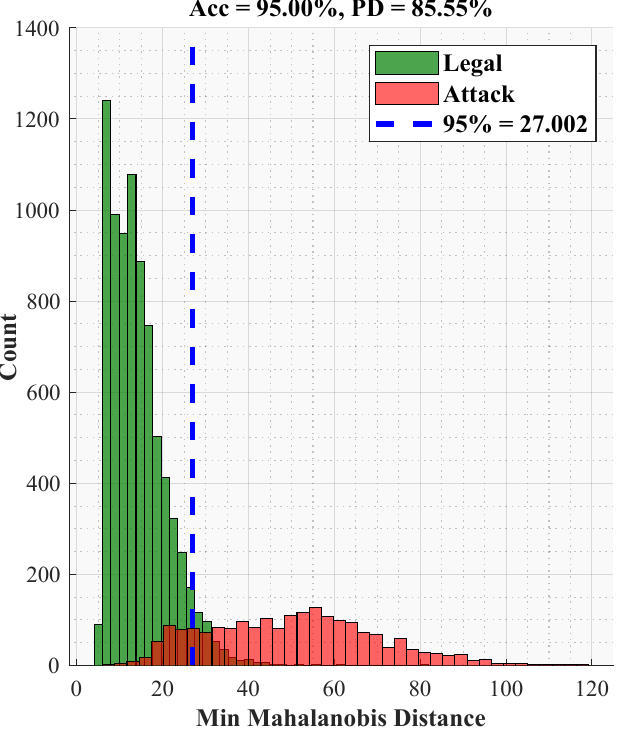}
	\hspace{-0.012\textwidth}
	\includegraphics[width=0.244\textwidth]{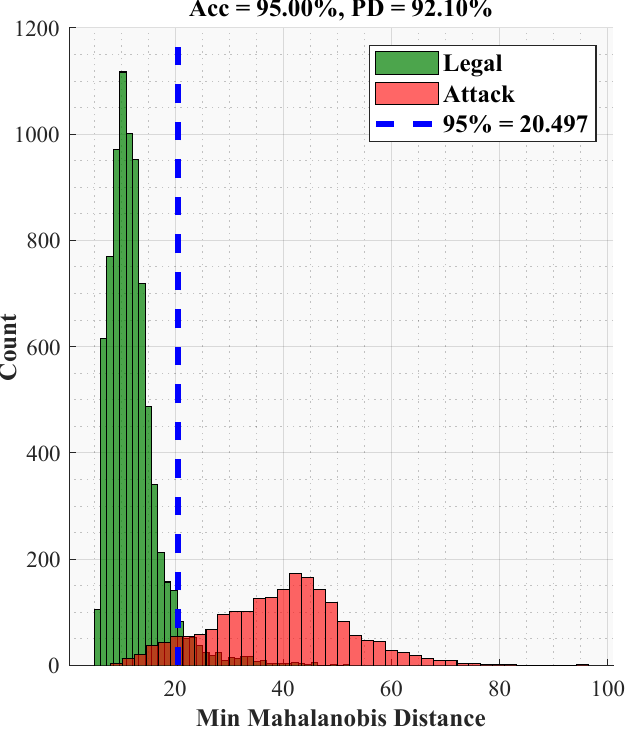}
	
	\vspace{0.5em}
	
	\makebox[0.244\textwidth]{\small (a) Baseline-MD verifier}
	\hspace{-0.012\textwidth}
	\makebox[0.244\textwidth]{\small (b) Proposed EAFM-MD verifier}
	
	\caption{OOD detection performance comparison: (a) the Baseline-MD verifier and (b) the proposed EAFM-MD verifier with $K=-10$~dB and SNR$=10$~dB. The vertical dashed line indicates the
		$p=0.95$ threshold for legal signal acceptance.}
	\label{fig:OOD_Detection_Performance}
\end{figure}

\begin{figure}[t]
	\centering
	\begin{minipage}[b]{0.425\textwidth}
		\centering
		\includegraphics[width=\textwidth]{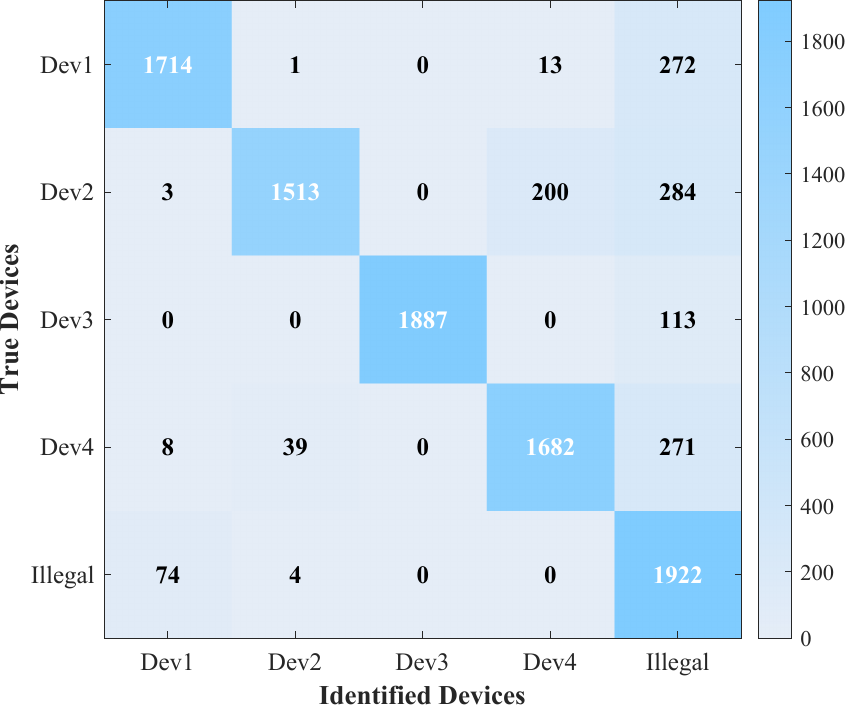}
		\par\vspace{0.5\baselineskip}
		\small (a) Proposed EAFM-MD verifier
		\par\vspace{0.5\baselineskip}
	\end{minipage}
	\hfill
	\begin{minipage}[b]{0.425\textwidth}
		\centering
		\includegraphics[width=\textwidth]{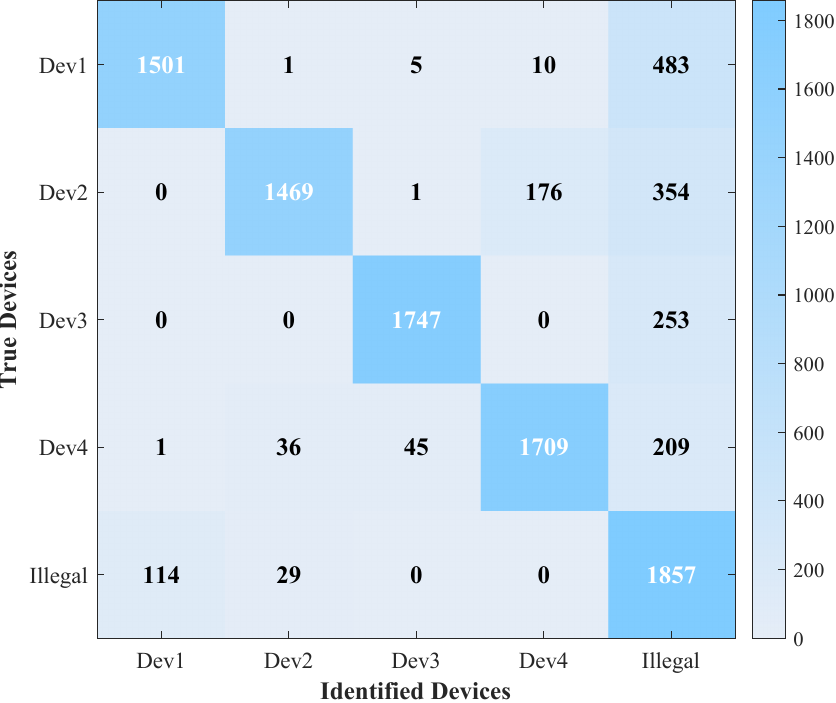}
		\par\vspace{0.5\baselineskip}
		\small (b) Baseline-MD verifier
		\par\vspace{0.5\baselineskip}
	\end{minipage}
	\caption{5-class confusion matrices of the proposed EAFM-MD verifier and the Baseline-MD verifier under the OOD condition with $K=-10$~dB and SNR$=10$~dB. }
	\label{fig:Confusion_Matrix}
\end{figure}

Table~\ref{tab:rician_k_factor_pla} presents the identification and verification performance of the EAFM-MD verifier and the Baseline-MD verifier under various Rician $K$ conditions. The proposed EAFM-MD verifier with the optimal $\alpha^*$ consistently demonstrates superior performance to the Baseline-MD verifier across all scenarios in both PD and Acc. For example, under ID(4\,dB) at 0\,dB SNR, the PD improves from 0.4425 to 0.8120 when using the EAFM-MD verifier. Under OOD($-5$\,dB) at 10\,dB SNR, the PD and Acc achieved by the EAFM-MD verifier are 0.9775 and 0.9020, respectively, while those achieved by the Baseline-MD verifier are 0.9520 and 0.8760. Even under the most challenging OOD($-10$\,dB) at 10\,dB SNR, the EAFM-MD verifier attains a PD of 0.9610 and Acc of 0.8495, surpassing the values of 0.9285 and 0.8033 obtained by the Baseline-MD verifier.

Fig.~\ref{fig:OOD_Detection_Performance} further compares the OOD detection performance under the same challenging condition with $K=-10$\,dB and SNR$=10$\,dB, where the threshold is set to achieve a 95\% legal signal acceptance rate. The histograms visualize the distribution of the minimum Mahalanobis distances for legitimate and attack samples. The proposed EAFM-MD verifier achieves a detection probability of 92.10\% for illegal attack devices, markedly outperforming the Baseline-MD verifier's 85.55\%, while maintaining the same legal acceptance rate of 95.00\%. Moreover, the 95\% threshold distance of the EAFM-MD verifier (20.497) is substantially lower than that of the Baseline-MD verifier (27.002), indicating that the EAFM-MD verifier produces more compact and separable feature representations, thereby facilitating a tighter decision boundary for legitimate devices.

Further corroboration comes from the confusion matrices in Fig.~\ref{fig:Confusion_Matrix} under the same OOD($-10$\,dB) condition at 10\,dB SNR, which visualize the detailed classification behavior. The proposed EAFM-MD verifier achieves notably fewer false alarms of legitimate devices being misclassified as unknown, reducing the total number from 1299 (Baseline-MD verifier) to 940 across all four legitimate devices. For instance, the false alarm rate for Dev1 drops from 24.15\% to 13.60\%, while that for Dev3 decreases from 12.65\% to 5.65\%. Meanwhile, the detection probability of the unknown attack device improves from 92.85\% (1857 out of 2000) to 96.10\% (1922 out of 2000), indicating that the EAFM-MD verifier more reliably isolates unknown transmitters in the feature space without increasing false alarms on legitimate devices.

The results show that the envelope-guided FiLM correction suppresses channel-induced variance in the features used for the fingerprint library, so that the Mahalanobis distance computed against per-device distributions is reduced and becomes more indicative of device identity.

\section{Conclusion}
In this paper, an EAFM identifier and an EAFM-MD verifier are proposed to address the cross-channel OOD challenge for SEI and to detect unknown spoofing devices for SEV in wireless fading channels. We first establish that the coefficient of variation of the signal envelope is strictly monotonic in the Rician-$K$ factor, which provides a unified theoretical foundation for both proposed methods. Leveraging this insight, we propose the EAFM identifier, a dual-branch neural architecture that adaptively modulates IQ-domain device features using FiLM-guided channel-conditioning information, thereby enhancing robustness against cross-channel variations. As an extension, the EAFM-MD verifier enables SEV by constructing a device fingerprint library using Mahalanobis distance metric learning and permits the detection of unknown spoofing devices without requiring adversarial samples during training. Numerical results demonstrate that the EAFM identifier achieves significantly higher accuracy retention rates across varying Rician-$K$ factors, while the EAFM-MD verifier achieves superior detection probability for unknown transmitters while maintaining high accuracy for legitimate devices. Both methods address the respective OOD challenges using only existing, legitimate signals and may thus be suitable for SEI and SEV in UAV and IIoT applications.

{\appendices \section{Proof of \textit{Lemma}~\ref{lemma1}} \label{appendix:fano_derivation}
	\begin{proof}	
		We first provide the necessary preliminaries of the Rician distribution. The Rician $K$-factor is defined as
		\begin{equation}
			K=\frac{\nu^2}{2\sigma^2},
			\label{eq:K-def}
		\end{equation}
		where $\nu$ denotes the amplitude of the LoS component and $\sigma^2$ represents the average power of the multipath components.
		
		The total average power is denoted by $\Omega$, which equals the second moment of the envelope amplitude, i.e.
		\begin{equation}
			\Omega=\mathbb{E}\!\left\{R^2\right\},
			\label{eq:Omega-def-1}
		\end{equation}
		where $R$ is the envelope amplitude. In the Rician model, $\Omega$ can also be written as
		\begin{equation}
			\Omega=\nu^2+2\sigma^2,
			\label{eq:Omega-def-2}
		\end{equation}
		which combined with \eqref{eq:K-def} yields
		\begin{equation}
			\Omega=2\sigma^2\left(1+K\right).
			\label{eq:Omega-def-3}
		\end{equation}
		
		The first and second moments of the Rician distribution are given respectively as
		\begin{align}
			\mathbb{E}\!\left\{R\right\} &= \sqrt{\frac{\pi \Omega}{4\left(1+K\right)}}\,L_{1/2}\!\left(-K\right), \label{eq:mean} \\
			\mathbb{E}\!\left\{R^2\right\} &= \Omega, \label{eq:second-moment}
		\end{align}
		where $L_{1/2}\left(x\right)$ is the generalized Laguerre function defined as
		\begin{equation}
			L_{1/2}\left(x\right)={}_1F_1\!\left(-\frac{1}{2};\,1;\,x\right),
			\label{eq:laguerre-def}
		\end{equation}
		where ${}_1F_1\left(a;\,b;\,x\right)$ is the confluent hypergeometric function of the first kind. Furthermore, we can get
		\begin{equation}
			L_{1/2}\left(-K\right)=e^{-K/2}\left[\left(1+K\right)I_0\!\left(\frac{K}{2}\right)+K\,I_1\!\left(\frac{K}{2}\right)\right],
			\label{eq:laguerre-bessel}
		\end{equation}
		where $I_0\left(\cdot\right)$ and $I_1\left(\cdot\right)$ denote the modified Bessel functions of the first kind of orders zero and one, respectively.
		
		The variance of the envelope amplitude is obtained as
		\begin{equation}
			\operatorname{Var}\!\left[R\right]=\mathbb{E}\!\left\{R^2\right\}-\left(\mathbb{E}\!\left\{R\right\}\right)^2.
			\label{eq:var-def}
		\end{equation}
		
		Substituting \eqref{eq:mean} and \eqref{eq:second-moment} into \eqref{eq:var-def}, the closed-form variance is given by
		\begin{equation}
			\operatorname{Var}\!\left[R\right]=\Omega\left[1-\frac{\pi}{4\left(1+K\right)}\left(L_{1/2}\left(-K\right)\right)^2\right].
			\label{eq:var}
		\end{equation}
		
		Therefore, the CV as a function of $K$ is obtained as
		\begin{equation}
			C_v\left(K\right)=\frac{\sqrt{\operatorname{Var}\!\left[R\right]}}{\mathbb{E}\!\left\{R\right\}}
			=\sqrt{\frac{4\left(1+K\right)}{\pi\left(L_{1/2}\left(-K\right)\right)^2}-1}.
			\label{eq:cv}
		\end{equation}
		
		We now proceed to prove the strict monotonicity. Define
		\begin{equation}
			L\left(K\right)=L_{1/2}\left(-K\right),
			\label{eq:L-def}
		\end{equation}
		and further define
		\begin{equation}
			g\left(K\right)=\frac{\left(L\left(K\right)\right)^2}{1+K}.
			\label{eq:g-def}
		\end{equation}
		
		Since the square root function is monotonically increasing on its domain, it suffices to show that $g\left(K\right)$ increases with $K$, or equivalently to prove that
		\begin{equation}
			\frac{\partial g\left(K\right)}{\partial K}>0.
			\label{eq:g-derivative-positive}
		\end{equation}
		
		Taking the derivative of $g\left(K\right)$ yields
		\begin{equation}
			\frac{\partial g\left(K\right)}{\partial K}=\frac{2L\left(K\right)\frac{\partial L\left(K\right)}{\partial K}\left(1+K\right)-\left(L\left(K\right)\right)^2}{\left(1+K\right)^2}.
			\label{eq:g-derivative-raw}
		\end{equation}
		
		Since $L\left(K\right)>0$ for $K>0$, where $K$ is understood as the linear Rician factor and therefore always positive, from \eqref{eq:g-derivative-raw} we see that $\frac{\partial g\left(K\right)}{\partial K}>0$ is equivalent to
		\begin{equation}
			2\frac{\partial L\left(K\right)}{\partial K}\left(1+K\right) > L\left(K\right).
			\label{eq:cond}
		\end{equation}
		
		Dividing both sides of \eqref{eq:cond} by $2L\left(K\right)\left(1+K\right)$ (which is positive) gives the equivalent condition, i.e.
		\begin{equation}
			\frac{1}{L\left(K\right)}\frac{\partial L\left(K\right)}{\partial K}>\frac{1}{2\left(1+K\right)}.
			\label{eq:cond-2}
		\end{equation}
		
		To utilize the Bessel function representations, we introduce an auxiliary function as
		\begin{equation}
			A\left(K\right)=e^{K/2}L\left(K\right).
			\label{eq:A-def}
		\end{equation}
		
		From \eqref{eq:laguerre-bessel}, we obtain
		\begin{equation}
			A\left(K\right)=\left(1+K\right)I_0\!\left(\frac{K}{2}\right)+K\,I_1\!\left(\frac{K}{2}\right).
			\label{eq:A-explicit}
		\end{equation}
		
		Then
		\begin{equation}
			\frac{1}{L\left(K\right)}\frac{\partial L\left(K\right)}{\partial K}=-\frac{1}{2}+\frac{1}{A\left(K\right)}\frac{\partial A\left(K\right)}{\partial K}.
			\label{eq:Lprime-over-L}
		\end{equation}
		
		Substituting \eqref{eq:Lprime-over-L} into \eqref{eq:cond-2} and simplifying, the inequality to be proven becomes
		\begin{equation}
			\frac{1}{A\left(K\right)}\frac{\partial A\left(K\right)}{\partial K}\ge \frac{K+2}{2\left(1+K\right)}.
			\label{eq:target}
		\end{equation}
		
		Computing $\frac{\partial A\left(K\right)}{\partial K}$ from \eqref{eq:A-explicit}, we obtain
		\begin{align}
			\frac{\partial A\left(K\right)}{\partial K} = \frac{\left(2+K\right)I_0\!\left(\frac{K}{2}\right) + \left(1+K\right)I_1\!\left(\frac{K}{2}\right)}{2}.
			\label{eq:Aprime_simplified}
		\end{align}
		
		Substituting \eqref{eq:Aprime_simplified} and \eqref{eq:A-explicit} into \eqref{eq:target}, and noting that all denominators are positive, we cross-multiply to obtain
		\begin{equation}
			\begin{aligned}
				&2\left(1+K\right)\left[\left(1+\frac{K}{2}\right)I_0\!\left(\frac{K}{2}\right)+\frac{1+K}{2}I_1\!\left(\frac{K}{2}\right)\right] \\
				&\quad\ge \left(K+2\right)\left[\left(1+K\right)I_0\!\left(\frac{K}{2}\right)+K\,I_1\!\left(\frac{K}{2}\right)\right].
			\end{aligned}
			\label{eq:cross-mult}
		\end{equation}
		
		Expanding and simplifying \eqref{eq:cross-mult}, the equivalent inequality is given by
		\begin{equation}
			\left[\left(1+K\right)^2-K\left(K+2\right)\right]\,I_1\!\left(\frac{K}{2}\right)=I_1\!\left(\frac{K}{2}\right)\ge 0.
			\label{eq:final-inequality}
		\end{equation}
		
		Since the modified Bessel function $I_1\left(x\right)$ is non-negative for all $x\ge 0$, we have $I_1\left(K/2\right)>0$ for $K>0$, and thus the inequality in \eqref{eq:final-inequality} holds strictly. Consequently, $\frac{\partial g\left(K\right)}{\partial K}>0$, which implies that $C_v\left(K\right)$ is strictly monotonically decreasing. This completes the proof.
		
	\end{proof}

\newpage

%
%
%
%

\vfill

\end{document}